%% file: neurips_GT.tex
\documentclass{article}

     \PassOptionsToPackage{numbers}{natbib}

\usepackage[final]{neurips_2021}

\usepackage[utf8]{inputenc} %
\usepackage[T1]{fontenc}    %
\usepackage{hyperref}       %
\usepackage{url}            %
\usepackage{booktabs}       %
\usepackage{amsfonts}       %
\usepackage{nicefrac}       %
\usepackage{microtype}      %
\usepackage{xcolor}         %
\usepackage{subfigure}

\input{macros.tex}

\usepackage[font=small]{caption}

\usepackage{algorithm}
\usepackage{algorithmic}

\newtheorem{lemma}{Lemma}
\newtheorem{proposition}[lemma]{Proposition}

\newtheorem{definition}{Definition}
\newtheorem{remark}[lemma]{Remark}
\newtheorem{assumption}{Assumption}
\newtheorem{theorem}[lemma]{Theorem}

\usepackage{pifont}%

\newcommand{\xmark}{\ding{55}}%

\usepackage[compact]{titlesec} 
\titlespacing*{\section}{14pt}{7pt}{4pt}
\titlespacing*{\subsection}{6pt}{3pt}{1pt}
\graphicspath{{figures/}}

  \usepackage[textwidth=6cm]{todonotes}

\title{An Improved Analysis of Gradient Tracking\\ for Decentralized Machine Learning}

\author{%
 Anastasia Koloskova\\
 EPFL\\
 \texttt{anastasia.koloskova@epfl.ch}
 \And
 Tao Lin\\
 EPFL\\
 \texttt{tao.lin@epfl.ch}
 \And
 Sebastian U. Stich\\
 EPFL\thanks{Current affiliation:\ CISPA Helmholtz Center for Information Security.}\\
 \texttt{sebastian.stich@epfl.ch}
}

\begin{document}

\maketitle
\input{main_GT.tex}

\end{document}

%% file: macros.tex
\usepackage{amsfonts}
\usepackage{amsmath}
\usepackage{amsthm}
\usepackage{amssymb}
\usepackage{dsfont}

\usepackage{xcolor}
\usepackage{color}
\usepackage{graphicx}

\usepackage{verbatim}
\usepackage{xspace} %

\usepackage{enumerate}
\usepackage{enumitem}

\providecommand{\lin}[1]{\ensuremath{\left\langle #1 \right\rangle}}
\providecommand{\abs}[1]{\left\lvert#1\right\rvert}
\providecommand{\norm}[1]{\left\lVert#1\right\rVert}

  \providecommand{\R}{\mathbb{R}} %

  \DeclareMathOperator{\E}{{\mathbb E}}
  \providecommand{\EE}[2]{{\mathbb E}_{#1}\left.#2\right. }  %

  \DeclareMathOperator*{\diag}{diag}

  \providecommand{\1}{\mathbf{1}}
  \renewcommand{\aa}{\mathbf{a}}
  \providecommand{\bb}{\mathbf{b}}

  \renewcommand{\gg}{\mathbf{g}}

  \providecommand{\uu}{\mathbf{u}}
  \providecommand{\vv}{\mathbf{v}}
  
  \providecommand{\xx}{\mathbf{x}}
  \providecommand{\yy}{\mathbf{y}}
  \providecommand{\zz}{\mathbf{z}}

  \providecommand{\cD}{\mathcal{D}}

  \providecommand{\cN}{\mathcal{N}}
  \providecommand{\cO}{\mathcal{O}}

  \providecommand{\mmu}{\boldsymbol{\mu}}
  
  \usepackage{bm}

\usepackage{url}

%% file: main_GT.tex
\begin{abstract}
We consider decentralized machine learning over a network where the training data is distributed across $n$ agents, each of which can compute stochastic model updates on their local data. The agent's common goal is to find a model that minimizes the average of all local loss functions.
While gradient tracking (GT) algorithms can overcome a key challenge, namely accounting for differences between workers' local data distributions, the known convergence rates for GT algorithms are not optimal with respect to
their dependence on the mixing parameter $p$ (related to the spectral gap of the connectivity matrix).
\\%
We provide a tighter analysis of the GT method in the stochastic strongly convex, convex and non-convex settings.
We improve the dependency on $p$ from $\cO(p^{-2})$ to $\cO(p^{-1}c^{-1})$ in the noiseless case and from $\cO(p^{-3/2})$ to $\cO(p^{-1/2}c^{-1})$ in the general stochastic case, where $c \geq p$ is related to the negative eigenvalues of the connectivity matrix (and is a constant in most practical applications). 
This improvement was possible due to a new proof technique which could be of independent interest.

\end{abstract}

\section{Introduction}
Methods that train machine learning models on decentralized  data offer many advantages over traditional centralized approaches in core aspects such as data ownership, privacy, fault tolerance and scalability~\cite{Kairouz2019:federated,Nedic2020:survey}.
Many current efforts in this direction come under the banner of federated learning~\cite{konecny2016federated2,McMahan16:FedLearning,McMahan:2017fedAvg,Kairouz2019:federated}, where a central entity orchestrates the training and collects aggregate updates from the participating devices.
Fully decentralized methods, that do not rely on a central coordinator and that communicate only with neighbors in an arbitrary communication topology, are still in their infancy~\cite{Lian2017:decentralizedSGD,Kong2021:consensus}.

The work of \citet{Lian2017:decentralizedSGD} on decentralized stochastic gradient descent (D-SGD)
has spurred the
research on decentralized training methods for machine learning models.
This lead to improved theoretical analyses~\cite{koloskova2020unified} and to improved practical schemes, such as support for time-varying topologies~\cite{nedic2014distributed,Assran:2018sdggradpush,koloskova2020unified} and methods with communication compression~\cite{Tang2018:decentralized,Wang2019:matcha,Koloskova:2019choco,Tang2019:squeeze}.
One of the most challenging aspect when training over decentralized data is data-heterogeneity, i.e.\ training data that is in a non-IID fashion distributed over the devices (for instance in data-center training) or generated in non-IID fashion on client devices~\cite{li2018fedprox,Kgoogle:cofefe,li2020feddane,Li2020:fedavg}.
For example, the D-SGD method has been shown to be affected by the heterogenity~\cite{koloskova2020unified}.

In contrast,  certain methods can mitigate the impact of heterogeneous data in decentralized optimization. For instance the \emph{gradient tracking} (GT) methods developed by  \citet{Lorenzo2016GT-first-paper} and \citet{Nedic2016DIGing}, or the later D${}^2$ method by
\citet{Tang2018:d2} %
which is designed for communication typologies that remain fixed and do not change over time.

It is well known that GT methods do not depend on the heterogeneity of the data and that they converge linearly on distributed strongly convex problem instances without stochastic noise~\cite{Lorenzo2016GT-first-paper,Nedic2016DIGing}. However,  when we apply these methods in the context of machine learning, we need to understand how they are affected by stochastic noise and how they behave on non-convex tasks. 

In this paper, we develop a new, and improved, analysis of the gradient tracking algorithm with a novel proof technique.
Along with the parallel contribution~\cite{Yuan2021d2-exact-diff-rates} that developed a tighter analysis of the D${}^2$ algorithm, we now have a more accurate understanding of in which setting GT works well and in which ones it does not, and our results allow for a more detailed comparison between the D-SGD, GT and D${}^2$ methods (see Section~\ref{sec:discussion} below).

Our analysis improves over all existing results that analyze the GT algorithm. Specifically, we prove a weaker dependence on the connectivity of the network (spectral gap) which is commonly incorporated into the convergence rates via the standard parameter $p$. 
For example, in the strongly convex setting with stochastic noise we prove that GT converges at the rate $ \tilde\cO\bigl(\frac{\sigma^2}{n\varepsilon } +  \frac{1}{c}\cdot \bigl(\frac{ \sigma}{\sqrt{p \varepsilon}}  + \frac{1}{p} \log \frac{1}{\varepsilon} \bigr)\bigr) $
where $\sigma^2$ is an upper bound on the variance of the stochastic noise, and $c \geq p$ a new parameter (often a constant). 
By comparing this result with the previously best known upper bound,
$\tilde\cO\bigl(\frac{\sigma^2}{n\varepsilon } +  \frac{1}{p}\cdot \bigl( \frac{ \sigma}{\sqrt{p\varepsilon}} + \frac{1}{p}\log \frac{1}{\varepsilon} \bigr)\bigr)$,
 by~\citet{pu2020gradient-tracking},
 we see that our upper bound improves the last two terms by a factor of $\smash{\frac{c}{p}} \geq 1$ and that the first term matches with known lower bounds~\cite{nemirovskyyudin1983}. The D${}^2$ algorithm~\cite{Tang2018:d2} only converges under the assumption that $c$ is a constant\footnote{In D${}^2$ the smallest eigenvalue  of the mixing matrix $W$ must bounded from  below: $\min_{i} \lambda_i(W)\geq -\tfrac{1}{3}$.} and the recent upper bound from~\cite{Yuan2021d2-exact-diff-rates} coincides with our worst case complexity for GT on all topologies where D${}^2$ can be applied.
We provide additional comparison of GT convergence rates in the Tables~\ref{tab:gt-str-conv} and \ref{tab:gt-non-conv}.

\begin{table*}[t]
	\begin{minipage}{\textwidth}
		\caption{Important advances for Gradient Tracking in the strongly convex case. Our analysis improves upon all prior rates for both with and without the stochastic noise in terms of the graph parameter $p$. 
		}
		\centering
		\label{tab:gt-str-conv}
		\resizebox{0.85\linewidth}{!}{
			\begin{tabular}{llcc}\toprule[1pt]
				Reference & {rate of convergence to $\epsilon$-accuracy}  & considered stochastic noise \\ \midrule
				\citet{Nedic2016DIGing} & $\cO\left(\dfrac{L^{3}}{\mu^{3} p^2 } \log \frac{1}{\varepsilon} \right)$ & \xmark\\
				\citet{Sulaiman2019} & $\cO\left( \dfrac{L}{\mu} \log \frac{1}{\varepsilon}+ \dfrac{1}{p^2} \log \frac{1}{\varepsilon} \right)$ & \xmark \\
				\citet{Qu2017GT-Harnessing-Smoothness} & $\cO\left(\dfrac{L^{2}}{\mu^{2} p^2} \log \frac{1}{\varepsilon} \right)$ & \xmark\\
				\citet{pu2020gradient-tracking} & $ \tilde\cO\left(\dfrac{\sigma^2}{\mu n \varepsilon} + \dfrac{\sqrt{L} \sigma}{\mu \sqrt{p} p \sqrt{\varepsilon}} + \dfrac{C_1}{\sqrt{\varepsilon}}\right)  $\footnote{$C_1$ is a constant that is independent of $\varepsilon$, but can depend on other parameters, such as $\sigma, \mu, L, p$} & \checkmark\\ \midrule
				this work & $ \tilde\cO\left(\dfrac{\sigma^2}{\mu n \varepsilon} +  \dfrac{\sqrt{L}  \sigma}{\mu \sqrt{p} c \sqrt{\varepsilon}}+  \dfrac{L}{\mu p c } \log \frac{1}{\varepsilon} \right)  $ & \checkmark
				\\ \bottomrule[1pt]
		\end{tabular}}
	\end{minipage}
\end{table*}

\textbf{Contributions.} Our main contributions can be summarized as:
\begin{itemize}[nosep,leftmargin=12pt]
	\item We prove better complexity estimates for the GT algorithm than known before with a new proof technique (which might be of independent interest).
	\item In the non-asymptotic regime (of importance in practice), the convergence rate depends on the network topology. By defining new graph parameters, we can give a tighter description of this dependency, explaining why the worst case behavior is rarely observed in practice (see Section~\ref{sec:parameterc}). We verify this dependence in numerical experiments.
	\item We show that in the presence of stochastic noise, the leading term in the convergence rate of GT is optimal---we are the first to derive this in the non-convex setting---and matching the unimprovable rate of all-reduce mini-batch SGD.
\end{itemize}

\section{Related Work}
\paragraph{Decentralized Optimization.}
Decentralized optimization methods have been studied for decades in the optimization and control community~\cite{Tsitsiklis1985:gossip,Nedic2009:distributedsubgrad,%
Wei2012:distributedadmm,%
Duchi2012:distributeddualaveragig}.
Many decentralized optimization methods \cite{Nedic2009:distributedsubgrad,Johansson2010:distributedsubgrad} are based on gossip averaging \cite{Kempe2003:gossip, Xiao2014:averaging,Boyd2006:randgossip}. Such methods usually also work well on non-convex problems and can be used used for training deep neural networks~\cite{Assran:2018sdggradpush,Lian2017:decentralizedSGD,Tang2018:d2}. 
There exists other methods, such as based on alternating direction method of multipliers (ADMM) \cite{Wei2012:distributedadmm,Iutzeler2013:randomizedadmm}, dual averaging~\cite{Duchi2012:distributeddualaveragig,Nedic2015:dualavg,Rabbat2015:mirrordescent}, primal-dual methods~\cite{Alghunaim2019:pd,Kovalev2021:free}, block-coordinate methods for generalized linear models~\cite{cola2018nips} or using new gradient propagation mechanisms~\cite{Vogels2021:relay}.

\paragraph{Decentralized Optimization with Heterogeneous Objective Functions.} 
There exists several algorithms that are agnostic to data-heterogeneity. Notably,  EXTRA~\cite{shi2015extra} and decentralized primal-dual gradient methods~\cite{Alghunaim2019:pd} do not depend on the data heterogeneity and achieve linear convergence in the strongly convex noiseless setting. However, these algorithms are not designed to be used for non-convex tasks.

D${}^2$ \cite{Tang2018:d2, Yuan2021d2-exact-diff-rates} (also known as exact diffusion \cite{yuan2019exact-diff-1,yuan2019exact-diff-2}) and Gradient Tracking (GT) \cite{Lorenzo2016GT-first-paper} (also known as NEXT~\cite{Lorenzo2016GT-first-paper} or DIGing \cite{Nedic2016DIGing}) are both algorithms that are agnostic to the data heterogeneity level, can tolerate the stochastic noise, and that can be applied to non-convex objectives such as the training of deep neural networks in machine learning. A limitation of the D${}^2$ algorithm is that it is not clear how it can be applied to  time-varying topologies, %
and that it can only be used on constant mixing topologies with negative eigenvalue bounded from below by $-\tfrac{1}{3}$. 
Other authors proposed algorithms that perform well on heterogeneous DL tasks~\cite{lin2021quasiglobal,yuan2021decentlam}, but theoretical proofs that these algorithms are independent of the degree of heterogeneity are still pending.

\paragraph{Gradient Tracking.}
There is a vast literature on the Gradient Tracking method itself. A tracking mechanism  was used by \citet{Zhu2010TrackingIdea} as a way to track the average of a distributed continuous process. \citet{Lorenzo2016GT-first-paper} applied this technique to track the gradients, and analyzed its asymptotic behavior in the non-convex setting with a time-varying topologies. \citet{Nedic2016DIGing} analyze GT (named as DIGing) in the strongly convex noiseless case with a time-varying network. \citet{Qu2017GT-Harnessing-Smoothness} extend the GT analysis to the non-convex, weakly-convex and strongly convex case without stochastic noise. \citet{Nedic2017GT-different-stepsizes} allow the different stepsizes on different workers. \citet{yuan2020GT-last-term} analyze asymptotic behavior of GT for dynamic optimization. \citet{pu2020gradient-tracking} studied the GT method on stochastic problems and strongly convex objectives. Further, \citet{Xin2019GT-stoch} analyze asymptotic behavior of GT with stochastic noise. For non-convex stochastic functions GT was analyzed by \citet{zhang2020GT-non-convex} and \citet{Lu2019GT-non-convex}.  
\citet{Li20GTandVR} combine GT with variance reduction to achieve linear convergence in the stochastic case. 
\citet{Tziotis2020GTsecond-order} obtain second order guarantees for GT.

\begin{table*}[t]
	\begin{minipage}{\textwidth}
		\caption{Important advances for Gradient Tracking in the non-convex case. Our result improves upon all existing rates in terms of the graph parameter $p$. %
		}
		\centering
		\label{tab:gt-non-conv}
		\resizebox{0.83\linewidth}{!}{
			\begin{tabular}{llc}\toprule[1pt]
				Reference & {rate of convergence to $\epsilon$-accuracy} & considered stochastic noise\\ \midrule
				\citet{Lorenzo2016GT-first-paper} & asymptotic convergence guarantees & \xmark \\
				\citet{zhang2020GT-non-convex} & $\cO \left( \dfrac{ L n \sigma^2 }{\varepsilon^2} + \dfrac{ L n }{ p^3 \varepsilon} \right)$ & \checkmark \\
				\citet{Lu2019GT-non-convex} &  $\cO \left(  \dfrac{C_1 + C_2\sigma}{\varepsilon^2 } \right)$\footnote{$C_1$ and $C_2$ are constants that are independent of $\varepsilon$, but can depend on other parameters, such as $\sigma, n, L, p$.} & \checkmark\\
				\midrule
				this work & $\tilde \cO \left(\dfrac{L\sigma^2 }{ {n} \varepsilon^2 }  + \dfrac{L \sigma}{(\sqrt{p} c + p\sqrt{n} ) \varepsilon^{\nicefrac{3}{2}}} + \dfrac{L}{p c \varepsilon} \right)$ & \checkmark
				\\ \bottomrule[1pt]
		\end{tabular}}
	\end{minipage}
\end{table*}

\section{Setup}
We consider optimization problems where the objective function is distributed across $n$ nodes, 
\begin{align}
\min_{\xx \in \R^d} \left[f(\xx) := \frac{1}{n}\sum_{i = 1}^{n} \big[f_i(\xx)= \E_{\xi \sim \cD_i} F_i(\xx,\xi)\big]\right]\,, \label{eq:problem}
\end{align}
where $f_i \colon \R^d \to \R$ denotes the local function available to the node $i$, $i \in [n] := \{1,\dots n\}$. Each $f_i$ is a stochastic function $f_i(\xx) = \E_{\xi \sim \cD_i} F_i(\xx, \xi)$ with access only to stochastic gradients $\nabla F_i(\xx, \xi)$. This setting covers empirical risk minimization problems with $\cD_i$ being a uniform distribution over the local training dataset. It also covers deterministic optimization when $F_i(\xx, \xi) = f_i(\xx)$, $\forall \xi$. 

We consider optimization over a decentralized network, i.e. when there is an underlying communication graph $G = (V, E)$, $|V| = n$, each of the nodes (e.g.\ a connected device) can communicate only  along the edges $E$. In decentralized optimization it is convenient to parameterize communication by a mixing matrix $W \in \R^{n \times n}$, where $w_{ij} = 0$ if and only if nodes $i$ and $j$ are not communicating, $(i, j) \notin E$. 

\begin{definition}[Mixing Matrix]\label{def:W} A matrix with non-negative entries $W \!\in\! [0,1]^{n \times n}$ that is symmetric  ($W\!=\!W^\top$) and doubly stochastic ($W\1 \!=\! \1$, $\1^\top W\! = \!\1^\top\!$), where $\1$ denotes the all-one vector in $\R^n$.
\end{definition}

\subsection{Notation}
We use the notation $\xx_i^{(t)} \in \R^d$, $\yy_i^{(t)} \in \R^d$ to denote the iterates and the tracking sequence, respectively, on node $i$ at time step $t$. For vectors $\zz_i \in \R^d$ ($\zz_i$ could for instance be $\xx_i^{(t)}$ or $\yy_i^{(t)}$) defined for $i \in [n]$ we denote by $\bar \zz = \frac{1}{n} \sum_{i = 1}^n \zz_i$.  

We use both vector and matrix notation whenever it is more convenient. 
For vectors $\zz_i \in \R^d$ defined  for $i \in [n]$ we denote by a capital letter the matrix with columns $\zz_i$, formally
\begin{align}\textstyle
Z := \left[ \zz_1,\dots, \zz_n\right] \in \R^{d\times n}\,, && \bar Z := \left[ \bar \zz,\dots, \bar \zz\right] \equiv Z\tfrac{1}{n} \1\1^\top\,, && \Delta Z := Z - \bar Z \,.
\end{align}
We extend this definition to gradients of~\eqref{eq:problem}, with $\nabla F(X^{(t )}, \xi^{(t )}), \nabla f(X^{(t)}) \in \R^{d\times n}$:
\begin{align*}
\nabla F(X^{(t )}, \xi^{(t )}) &:= \left[\nabla F_1(\xx_{1}^{(t)}, \xi_1^{(t)}), \dots,  \nabla F_n(\xx_{n}^{(t)}, \xi_n^{(t)})\right] \,, \\ %
 \nabla f(X^{(t)}) &:= \left[ \nabla f(\xx_1^{(t)}) , \dots,   \nabla f(\xx_n^{(t)})  \right] \,.%
\end{align*}

\subsection{Algorithm}
The Gradient Tracking algorithm (or NEXT, DIGing) can be written as
\begin{align}\label{eq:GT-matrix}
\begin{pmatrix}
X^{(t + 1)}\\
\gamma Y^{(t + 1)}
\end{pmatrix}^\top =  \begin{pmatrix}
	X^{(t)}\\
	\gamma Y^{(t)}
	\end{pmatrix}^\top \begin{pmatrix}
	W & 0 \\
	- W & W 
	\end{pmatrix} + \gamma \begin{pmatrix}
	0\\
	\nabla F(X^{(t + 1)}, \xi^{(t + 1)}) - \nabla F(X^{(t)}, \xi^{(t)})
	\end{pmatrix}^\top \tag{GT}
\end{align}
in matrix notation. Here
and $X^{(t)} \in \R^{d \times n}$ denotes the iterates, $Y^{(t)} \in \R^{d \times n}$, with $Y^{(0)}= \nabla F(X^{(t)}, \xi^{(t)})$ the sequence of tracking variables, and $\gamma >0$ denotes the stepsize. 
This update is summarized in Algorithm~\ref{alg:gt}. 

\begin{algorithm}[ht]
	\caption{\textsc{Gradient Tracking}}\label{alg:gt}
	\let\oldfor\algorithmicfor
	\renewcommand{\algorithmicfor}{\textbf{in parallel on all workers $i \in [n]$, for}}
	\let\oldendfor\algorithmicendfor
	\renewcommand{\algorithmicendfor}{\algorithmicend\ \textbf{parallel for}}
	\begin{algorithmic}[1]
		\INPUT{Initial values $\xx_i^{(0)} \in \R^d$ on each node $i \in [n]$, communication graph $G = ([n], E)$ and mixing matrix $W$, stepsize $\gamma$, initialize $\yy_i^{(0)} = \nabla F_i(\xx_i^{(0)}, \xi_i^{(0)})$, $\gg_i^{(0)} = \yy_i^{(0)}$ in parallel for $i \in [n]$.}\\[1ex]
		\FOR{$t=0,\dots, T-1$}
		\STATE each node $i$ sends $\Big(\xx_i^{(t)},\yy_i^{(t)}\Big)$ to is neighbors
		\STATE $\xx_i^{(t + 1)} = \sum_{j: \{i, j\} \in E } w_{ij} \Big(\xx_j^{(t)}   - \gamma \yy_j^{(t)} \Big)$ \hfill $\triangleright$ update model parameters
		\STATE Sample $\xi_i^{(t + 1)}$, compute gradient $\gg_i^{(t + 1)} = \nabla F_i\Big(\xx_i^{(t + 1)}, \xi_i^{(t + 1)}\Big)$
		\STATE  $\yy_i^{(t + 1)} = \sum_{j: \{i, j\} \in E } w_{ij} \yy_j^{(t)} + \Big(\gg_i^{(t + 1)} - \gg_i^{(t)}\Big)$ \hfill $\triangleright$ update tracking variable
		\ENDFOR
	\end{algorithmic}
\end{algorithm}
Each node $i$ stores and updates two variables, the model parameter $\xx_i^{(t)}$ and the tracking variable $\yy_i^{(t)}$. The model parameters are updated on line~3 with a decentralized SGD update but using $\yy_i^{(t)}$ instead of a gradient. Variable $\yy_i^{(t)}$ tracks the average of all local gradients on line 5. Intuitively, the algorithm is agnostic to the functions heterogeneity because $\yy_i^{(t)}$ is `close' to the full gradient of $f(\xx)$ (suppose we would replace line 5 with exact averaging in every timestep, then $\yy_i^{(t+1)}=\frac{1}{n} \sum_{i = 1}^n \gg_i^{(t+1)}$. For further discussion of the tracking mechanism refer to
\cite{Lorenzo2016GT-first-paper,Nedic2016DIGing,pu2020gradient-tracking}.

\begin{table}[tb]
\centering
	\begin{tabular}{lll} 
	\toprule
		graph/topology & $1/p$ & $c$\\ \midrule
		ring & $\cO (n^2)$ & $\nicefrac{8}{9}$\\
		2d-torus & $\cO (n)$ & $\geq \nicefrac{4}{5}$ \\
		fully connected & $\cO(1)$ & $1$ \\ \bottomrule
	\end{tabular}		
	\vspace{2mm}
	\caption{Parameters $p$ and $c$ for some common network topologies  on $n$ nodes  for uniformly averaging $W$, i.e. $w_{ij} = \frac{1}{deg(i)} = \frac{1}{deg(j)}$ for $\{i,j\} \in E$, see e.g. \cite{Nedic2018:graphs}.}
	\label{tab:pc}
\end{table}

\subsection{Assumptions}
We first state an assumption on the mixing matrix.
\begin{assumption}[Mixing Matrix]\label{a:W} 
	Let $\lambda_i(W)$, $i \in [n]$, denote the eigenvalues of the mixing matrix $W$ with
	$
	1 = \lambda_1(W) > \lambda_2(W) \geq \dots \geq \lambda_n(W) > -1.
	$
	With this, we can define the spectral gap $\delta = 1  - \max\{|\lambda_2(W)|,|\lambda_n(W)| \}$, and the mixing parameters
	\begin{align}
	p = 1  - \max\{|\lambda_2(W)|,|\lambda_n(W)| \}^2\,, && c = 1 - \min\{\lambda_n(W),0\}^2\,. \label{def:p}
	\end{align}
	We assume that $p > 0$ (and consequently $c>0$).
\end{assumption}
The assumption $p>0$ ensures that the network topology is connected, and that the consensus distance decreases linearly after each averaging step, i.e.\
$
\norm{XW - \bar X}_F^2 \leq (1 - p) \norm{X - \bar X}_F^2\,, \forall X \in \R^{d \times n}.
$
The parameter $p$ is closely related to the spectral gap $\delta$ as it holds $p = 2 \delta - \delta^2$. From this we can conclude that $\delta \leq p \leq 2 \delta$ and, asymptotically for $\delta \to 0$, $p \to 2\delta$. 
Assuming a lower bound on $p$ (or equivalently $\delta$) is a standard assumption in the literature. 

The parameter $c$ is related to the most negative eigenvalue. From the definition~\eqref{def:p} it follows that the auxiliary mixing parameter $c \geq p$ for all mixing matrices $W$. The parameters $p$ and $c$ are only equal when $|\lambda_n(W)| \geq |\lambda_2(W)|$ and $\lambda_n(W)\leq 0$.
Moreover, if  the diagonal entries $w_{ii}$ (self-weights) of the mixing matrix are all strictly positive, then $c$ has to be strictly positive. 

\begin{remark}[Lower bound on $c$.]
	Let $W$ be a mixing matrix with diagonal entries (self-weights) $w_{ii} \geq \rho > 0$, for a parameter $\rho$. Then $\lambda_n (W) \geq 2\rho - 1$ and $c \geq \min\{ 2\rho, 1\}$. 
\end{remark}
This follows from Gershgorin's circle theorem \cite{Gerschgorin31} that guarantees $\lambda_n(W) \geq 2\rho-1$, and hence $c \geq 1-\min\{2\rho-1,0\}^2 \geq \min\{ 2\rho, 1\}$.

For many choices of  $W$ considered in practice, most notably when the graph $G$ has constant node-degree and the weights $w_{ij}$ are chosen by the popular Metropolis-Hastings rule, i.e.\ $w_{ij}=w_{ji} = \min\bigl\{\frac{1}{\deg(i)+1},\frac{1}{\deg(j)+1}\bigr\}$ for $(i,j) \in E$, $w_{ii} = 1 - \sum_{j = 1}^n w_{ij} \geq \frac{1}{\max_{j \in [n]} \deg(j)}$, see also \cite{Xiao2014:averaging,Boyd2006:randgossip}.
In this case, the parameter $c$ can be bounded by a constant depending on the maximal degree. Moreover, for any given $W$, considering $\frac{1}{2}(W+I_n)$ instead (i.e.\ increasing the self-weights), ensures that $c=1$.
However, in contrast to e.g.\ the analysis in~\cite{Yuan2021d2-exact-diff-rates} we do not need to pose an explicit bound on $c$ as an assumption.
In practice, for many graphs, the parameter $c$ is bounded by a constant (see Table~\ref{tab:pc}).

We further use the following standard assumptions:
\begin{assumption}[$L$-smoothness]\label{a:lsmooth_nc}
	Each function $f_i \colon \R^d \to \R$, $i \in [n]$
	is differentiable and there exists a constant $L \geq 0$ such that for each $\xx, \yy \in \R^d$:
	\begin{align}\label{eq:smooth_nc}
	&\norm{\nabla f_i(\yy) - \nabla f_i(\xx) } \leq L \norm{\xx -\yy}\,. %
	\end{align}
\end{assumption}

Sometimes we will in addition assume that the functions are (strongly) convex.
\begin{assumption}[$\mu$-strong convexity]\label{a:mu-convex}
	Each function $f_i \colon \R^d \to \R$, $i \in [n]$ is $\mu$-strongly convex for constant $\mu \geq 0$, i.e. for all $\xx, \yy \in \R^d$:
	\begin{align}\textstyle
	f_i(\xx) - f_i(\yy) + \frac{\mu}{2} \norm{\xx - \yy}_2^2 \leq \langle \nabla f_i(\xx), \xx - \yy \rangle\,. %
	\label{def:strongconvex}
	\end{align}
\end{assumption}

\begin{assumption}[Bounded noise]\label{a:opt_nc}
	We assume that there exists constant $\sigma $ s.t. $\forall \xx_1, \dots \xx_n \in \R^d$
	\begin{align} \textstyle
	\frac{1}{n} \sum_{i = 1}^n \EE{\xi_i}{\norm{\nabla F_i(\xx_i, \xi_i) - \nabla f_i(\xx_i)}}^2_2 \leq \sigma^2 \,. \label{eq:noise_opt_nc}
	\end{align}
\end{assumption}
We discuss possible relaxations of these assumptions in Section~\ref{sec:general} below.

\section{Convergence results}
\label{sec:results}
We now present our novel convergence results for GT in Section~\ref{sec:general} and Section~\ref{sec:consensus} below.
We provide a proof sketch to explain the key difficulties and technical novelty compared to prior results later in the next Section~\ref{sec:sketch}.

\subsection{Main theorem---GT convergence in the general case}
\label{sec:general}

\begin{theorem}\label{thm:GT-better-upper-bound}
	Let $\xx_i^{(t)}$, $i \in [n]$, $T > \frac{2}{p}\log\big(\frac{50}{p} (1+ \log \frac{1}{p}) \big)$ denote the iterates of the GT Algorithm~\ref{alg:gt} with a mixing matrix as in Definition~\ref{def:W}.	
	 If Assumptions~\ref{a:W}, \ref{a:lsmooth_nc} and \ref{a:opt_nc} hold, then  there exists a stepsize $ \gamma$ such that the optimization error is bounded as follows: \\
	\textbf{Non-convex:} Let $F_0 = f(\bar\xx^{(0)}) - f^\star$ for $f^\star \leq \min_{\xx \in \R^d} f(\xx)$. Then  it holds
	\begin{align*}
	{\textstyle \frac{1}{T + 1} \sum_{t = 0}^T \norm{\nabla f(\bar\xx^{(t)})}_2^2 } \leq \varepsilon\,, && \!\!\!\!\!\!\text{ after  }\!\!\!\!\!&& \tilde\cO\left( {\frac{\sigma^2}{n\varepsilon}} +  \frac{\sigma }{({\sqrt{p} c + p\sqrt{n}}) \varepsilon^{\nicefrac{3}{2}}} + \frac{1 + L \tilde R_0^2F_0^{-1} }{pc \varepsilon}\right) \!\cdot  \!  LF_0 \, &&\!\!\!\!\!\!\text{   iterations.}
	\end{align*}\\
	\textbf{Strongly-convex:}
	Under the additional Assumption \ref{a:mu-convex} with $\mu > 0$ and weights $w_t \geq 0$, $W_T=\sum_{t=0}^T w_t$, specified in the proof,  it holds for $R_{T+1}^2 = \norm{\bar \xx^{(T+1)} - \xx^\star}^2$:
	\begin{align*}%
	{\textstyle \sum_{t = 0}^T \frac{w_t}{W_T} \left[\E f(\bar\xx^{(t)}) - f^\star \right] + \frac{\mu}{2} R_{T + 1} \leq \varepsilon}\,, &&\!\! \text{after}\!\!&& \tilde\cO\left(\frac{\sigma^2}{\mu n \varepsilon} +  \frac{\sqrt{L}  \sigma}{\mu \sqrt{p} c \sqrt{\varepsilon}} + \frac{L}{\mu pc} \log \frac{1}{\varepsilon}\right)&& \!\!\!\text{iterations.}
	\end{align*}
	\textbf{General convex:}
		Under the additional Assumption \ref{a:mu-convex} with $\mu \geq 0$, it holds for $R_{0}^2 = \norm{\bar \xx^{(0)} - \xx^\star}^2$:
	\begin{align*}
	{\textstyle \frac{1}{T + 1}\sum_{t = 0}^T \left[\E f(\bar\xx^{(t)}) - f^\star \right] \leq \varepsilon}\,, &&\!\!\text{after}\!\!\!&&  \tilde\cO\left(\frac{\sigma^2}{n\varepsilon^2} +  \frac{\sqrt{L}\sigma }{\sqrt{p} c \varepsilon^{\nicefrac{3}{2}}}  + \frac{L (1 + \tilde R_0^2R_0^{-2}) }{pc \varepsilon} \right) \! \cdot \! R_0^2 \!\! && \text{iterations,}
	\end{align*}
	where $\tilde R_0^2 = \frac{1}{n} \sum_{i = 1}^n \|\xx_i^{(0)} - \bar \xx^{(0)}\|^2 + \frac{1}{nL^2} \sum_{i = 1}^n \|\yy_i^{(0)} - \bar \yy^{(0)}\|^2$.
\end{theorem}
From these results we see that the leading term in the convergence rate (assuming $ \sigma>0$) is not affected by the graph parameters.
 Moreover, in this term we see a linear speedup in $n$, the number of workers. The leading terms of all three results match with the convergence estimates for all-reduce mini-batch SGD~\cite{Dekel2012:minibatch,Stich19sgd} and is optimal~\cite{nemirovskyyudin1983}.
 This means, that after a sufficiently long transient time, GT achieves a linear speedup in $n$. This transient time depends on the graph parameters $p$ and $c$, but not on the data-dissimilarity. %
We will discuss the dependency of the convergence rate on the graph parameters $c,p$ more carefully below in Sections~\ref{sec:discussion} and~\ref{sec:experiments}, and compare the convergence rate to the convergence rates of D-SGD and D${}^2$.

\textbf{Possible Relaxations of the Assumptions.}
Before moving on to the proofs, we mention briefly a few possible relaxations of the assumptions that are possible with only slight adaptions of the proof framework. These extensions can be addressed with known techniques and are omitted for conciseness. We give here the necessary references for completeness. 
\begin{itemize}[nosep,leftmargin=12pt]
\item \textbf{Bounded Gradient Assumption I.} The uniform bound on the stochastic noise in Assumption~\ref{a:opt_nc} could be relaxed by allowing the noise to grow with the gradient norm~\cite[Assumption 3b]{koloskova2020unified}.
\item \textbf{Bounded Gradient Assumption II.} In the convex setting it has been observed that $\sigma^2$ can be replaced with $\sigma^2_\star := \frac{1}{n} \sum_{i = 1}^n \EE{\xi_i}{\norm{\nabla F_i(\xx^\star, \xi_i) - \nabla f_i(\xx^\star)}}^2_2 $, the noise at the optimum. However, this requires smoothness of each $F_i(\xx,\xi)$, $\xi \in \cD_i$, which is stronger than our Assumption~\ref{a:lsmooth_nc}. For the technique see e.g.~\cite{Nguyen2018:async}. %
\item \textbf{Different mixing for $X$ and $Y$.} In Algorithm~\ref{alg:gt}, both the $\xx$ and $\yy$ iterates are averaged on the same communication topology (the same mixing matrix). This can be relaxed by allowing for two separate matrices. This follows from inspecting our proof below.
\item \textbf{Local Steps.} It is possible to extend Algorithm~\ref{alg:gt} and our analysis in Theorem~\ref{thm:GT-better-upper-bound} to allow for local computation steps. Mixing matrix would alternate between identity matrix $I$ (no communication, local steps) and $W$ (communication steps). \\
However, it is non trivial to extend our analysis to the general time-varying graphs, as the product of two arbitrary mixing matrices $W_1 W_2$ might be non symmetric. 
\end{itemize}

\subsection{Faster convergence on consensus functions}
\label{sec:consensus}
We now state an additional result, which improves Theorem~\ref{thm:GT-better-upper-bound} on the consensus problem, defined as
\begin{align}
\min \left[f(\xx) = \frac{1}{n} \sum_{i=1}^n \big[ f_i(\xx):= \tfrac{1}{2}\norm{\xx- \mmu_i}^2 \big] \right] \,, \label{eq:consensus}
\end{align}
for vectors $\mmu_i \in \R^d$, $i \in [n]$ and optimal solution $\xx^\star = \frac{1}{n}\sum_{i=1}^n \mmu_i$. Note that this is a special case of the general problem~\eqref{eq:problem} without stochastic noise ($\sigma=0$).
For this function, we can improve the complexity estimate that would follow from Theorem~\ref{thm:GT-better-upper-bound} by proving a convergence rate that does not depend on $c$.

\begin{theorem}\label{thm:consensus}
Let $f$ be as in \eqref{eq:consensus} let Assumption~\ref{a:W} hold.
Then there exists a stepsize $\gamma \leq p$ such that
it holds $\frac{1}{n} \sum_{i = 1}^n \big\| \xx_i^{(T)} - \xx^\star \big\|^2 \leq \epsilon$, for the iterates GT~\ref{alg:gt} and any $\epsilon>0$, after at most
 $T= \tilde \cO \left( p  \log \frac{1}{\epsilon} \right)$
iterations.
\end{theorem}

\section{Discussion}
\label{sec:discussion}
We now provide a discussion of these results.

\subsection{Parameter $c$}
\label{sec:parameterc}

The convergence rate in Theorem~\ref{thm:GT-better-upper-bound} depends on the parameter $c$, that in the worst case could be as small as $p$. 
In this case our theoretical result does not improve over existing results for the strongly convex case. However, for many graphs in practice parameter $c$ is bounded by a constant (see Table~\ref{tab:pc} and discussion below Assumption~\ref{a:W}).

While we show in Theorem~\ref{thm:consensus} that it is possible to remove the dependency  on $c$ entirely from the convergence rate in special cases,  it is still an open question if the parameter $c$ in Theorem~\ref{thm:GT-better-upper-bound} is tight in general. 

\subsection{Comparison to prior GT literature}
Tables~\ref{tab:gt-str-conv} and \ref{tab:gt-non-conv} compare our theoretical convergence rates in strongly convex and non convex settings. Our result tightens all existing prior work.

\subsection{Comparison to other methods. }
We now compare our complexity estimate of GT to D-SGD and D${}^2$ in the strongly convex case. Analogous observations hold for the other cases too.

\textbf{Comparison to D-SGD.} A popular algorithm for decentralized optimization is D-SGD~\cite{Lian2017:decentralizedSGD} that converges as \cite{koloskova2020unified}: \vspace{-2mm}
\begin{align}
\tilde \cO\left( \frac{\sigma^2}{\mu n \varepsilon} + \frac{\sqrt{L} \left( \zeta + \sqrt{p}  \sigma\right)}{\mu p \sqrt{\varepsilon}} +   \frac{L}{\mu p } \log \frac{1}{\varepsilon}\right)\,. \tag{D-SGD} 
\end{align}
While GT is agnostic to data-heterogenity, here the convergence estimate depends on the data-heterogenity, measured by a constant $\zeta^2$ that satisfies:
\begin{align} \textstyle
\frac{1}{n} \sum_{i = 1}^n \norm{\nabla f_i(\xx^\star) - \nabla f(\xx^\star)}_2^2 \leq \zeta^2 \,. %
\label{eq:new-attention}
\end{align}
Comparing with Theorem~\ref{thm:GT-better-upper-bound}, GT completely removes dependence on data heterogeneity level $\zeta$. Moreover, even in the homogeneous case when $\zeta = 0$, GT enjoys the same rate as D-SGD for many practical graphs when $c$ is bounded by a constant. 

\textbf{Comparison to D${}^2$.}
Similarly to GT, D${}^2$ also removes the dependence on functions heterogeneity. The convergence rate of D${}^2$ holds under assumption that $\lambda_{\min}\left(W\right) > -\frac{1}{3}$ and it is equal to \cite{Yuan2021d2-exact-diff-rates}:
\begin{align}
\cO\left( \frac{\sigma^2}{\mu n \varepsilon} + \frac{\sqrt{L } \sigma}{\mu \sqrt{p} \sqrt{\varepsilon}} +   \frac{L}{\mu p} \log \frac{1}{\varepsilon} \right) \,.\tag{D${}^2$}  \label{eq:d2}
\end{align}
Under the assumption $\lambda_{\min}\left(W\right) > -\frac{1}{3}$ the parameter $c$ is a constant, and the GT rate estimated in Theorem~\ref{thm:GT-better-upper-bound} matches~\eqref{eq:d2}.

\section{Proof sketch of the main theorem}
\label{sec:sketch}
Here we give a proof sketch for Theorem~\ref{thm:GT-better-upper-bound}, for the special case of strongly convex objectives. We give all proof details in the appendix and highlight the main technical difficulties and novel techniques.

\textbf{Key Lemma.}
It is very common---and useful---to write the iterates in the form $X^{(t)} = \bar X^{(t)} + (X^{(t)} - \bar X^{(t)})$, where $\bar X^{(t)}$ denotes the matrix with the average over the nodes.
We can then separately analyze $\bar X^{(t)}$ and the consensus difference $\Delta X^{(t)}:=(X^{(t)} - \bar X^{(t)})$ (and $\Delta Y^{(t)}:=(Y^{(t)} - \bar Y^{(t)})$). Define $\tilde W = W - \frac{\1\1^\top}{n}$. From the update equation~\eqref{eq:GT-matrix} we see that 
\begin{align*}
\begin{pmatrix}
\Delta X^{(t + 1)}\\
\gamma \Delta Y^{(t + 1)}
\end{pmatrix}^{\!\top}\!\!\!\! = \underbrace{\begin{pmatrix}
	\Delta X^{(t)}\\
	\gamma \Delta Y^{(t)}
	\end{pmatrix}^{\!\top}\!\!\!}_{=:\Psi_t} \underbrace{\begin{pmatrix}
	\tilde W & 0\\
	- \tilde W & \tilde W 
	\end{pmatrix} }_{=: J}  \!+ \gamma \underbrace{\begin{pmatrix}
	0\\
	\left(\nabla F(X^{t + 1}, \xi^{t + 1}) - \nabla F(X^{t}, \xi^t)\right)(I - \frac{\1\1^\top}{n})
	\end{pmatrix}^{\!\top}}_{=: E_t} ,
\end{align*} 
in short, by using the notation $\Psi_t$, $J$, and $E_t$ as introduced above,
\begin{align}
 \Psi_{t+1} = \Psi_t J  + \gamma E_t\,. \label{eq:25}
\end{align}
We could immediately adapt the proof technique from~\cite{koloskova2020unified} if it would hold that the spectral radius of $J$ is smaller than one. However, this is not the case, and in general $\norm{J}>1$. 

Note that for any integer $i\geq 0$:
\begin{align}
 J^i &= \begin{pmatrix}
\tilde W^i & 0\\
- i \tilde W^i & \tilde W^i
\end{pmatrix} &
\|J^i\|^2  = \|\tilde W^{i}\|^2 + i^2 \|\tilde W^{i}\|^2 \leq (1-p)^i + i^2 (1-p)^i\,, \label{eq:i-p}
\end{align}
by Assumption~\ref{a:W}.
With this observation we can now formulate a key lemma:
\begin{lemma}[Contraction]
\label{lemma:key}
For any integer $\tau \geq \frac{2}{p}\log\left(\frac{50}{p} (1+ \log \frac{1}{p}) \right)$ it holds that   $\norm{J^\tau}^2 \leq \frac{1}{2}$\,.
\end{lemma}
While the constants in this lemma are chosen to ease the presentation, most important for us is that after $\tau=\tilde\Theta\bigl(\frac{1}{p}\bigr)$ communication rounds, old parameter values (from $\tau$ steps ago) get discounted and averaged by a constant factor.
We can alternatively write the statement of Lemma~\ref{lemma:key} as
\[
 \norm{Z J^\tau  - \bar Z}_F^2 \leq \tfrac{1}{2} \norm{Z-\bar Z}_F^2\,, \qquad \forall Z \in \R^{2d \times n}\,.
\]
This resembles \cite[Assumption 4]{koloskova2020unified}
and the proof now follows the same pattern. A few crucial differences remain, as the result in \cite{koloskova2020unified} depends on a data-dissimilarity parameter which we can avoid by carefully estimating the tracking errors. For completeness, we sketch the outline and give all details in the appendix. %

\textbf{Average Sequence.}
First, we consider the average sequences $\bar X^{(t)}$ and $\bar Y^{(t)}$. As all columns of these matrices are equal, we can equivalently consider a single column only: $\bar \xx^{(t)}$ and $\bar \yy^{(t)}$.
\begin{lemma}[Average]\label{lem:average}
	It holds that %
	\begin{align}
	\bar \yy^{(t)} = \frac{1}{n} \sum_{i = 1}^n \nabla F_i \big(\xx_i^{(t)}, \xi_i^{(t)} \big)\,, && \bar{\xx}^{(t + 1)} = \bar{\xx}^{(t)} - \gamma \frac{1}{n} \sum_{i = 1}^n \nabla F_i \big(\xx_i^{(t)}, \xi_i^{(t)} \big) \,. \label{eq:23}
	\end{align}
\end{lemma}
This follows directly from the update~\eqref{eq:GT-matrix} and the fact that $\bar X = \bar X W$ for doubly stochastic mixing matrices.
The update of $\bar \xx^{(t)}$ in~\eqref{eq:23} is almost identical to one step of mini-batch SGD (on a complete graph). 
The average sequence behaves almost as a SGD sequence:
\begin{lemma}[Descent lemma, {\cite[Lemma 8]{koloskova2020unified}}]\label{lem:descent}
	Under the Assumptions of Theorem~\ref{thm:GT-better-upper-bound} for the convex functions,
	the averages $\bar{\xx}^{(t)} := \frac{1}{n}\sum_{i=1}^n \xx_i^{(t)}$ of the iterates of Algorithm \ref{alg:gt} with the stepsize $\gamma \leq \frac{1}{12 L} $ satisfy %
	\begin{equation}
	\resizebox{0.93\linewidth}{!}{
	$ \displaystyle
	\E \big \|\bar{\xx}^{(t + 1)} - \xx^\star \big \|^2 \leq \left(1 - \dfrac{\gamma\mu}{2}\right) \E { \big \|\bar{\xx}^{(t)} - \xx^\star \big\|}^2 + \dfrac{\gamma^2\sigma^2}{n} - \gamma e_t+  \dfrac{3 \gamma L }{n} \sum_{i = 1}^{n} \E  \big \|\bar{\xx}^{(t)} - \xx_i^{(t)} \big \|^2,    
	$
	} \label{eq:x}
	\end{equation}
	where $e_t = \E f(\bar{\xx}^{(t)}) - f^\star $, for $f^\star = \min_{\xx \in \R^d} f(\xx)$. %
\end{lemma}

\textbf{Consensus Distance.}
The main difficulty comes from estimating the consensus distance 
$\norm{\Psi_{t}}^2$, in the notation introduced in~\eqref{eq:25}.
Note that
\begin{align*}
 \|\Psi_t\|^2 = \frac{1}{n} \sum_{i = 1}^n \big\|\xx_i^{(t)} - \bar \xx^{(t)} \big \|_2^2 +  \frac{\gamma^2}{n} \sum_{i = 1}^n \big \|\yy_i^{(t)} - \bar \yy^{(t)} \big \|_2^2 \,.
\end{align*}
By unrolling~\eqref{eq:25} for $\tau \leq  k \leq  2 \tau$, $\tau = \frac{2}{p}\log\left(\frac{50}{p} (1+ \log \frac{1}{p}) \right) + 1$ steps, \vspace{-2mm}
\begin{align}\label{eq:mid-proof}
\Psi_{t + k} = \Psi_t J^k + \gamma \sum_{j = 1}^{k - 1}  E_{t + j - 1}J^{k - j}\,.
\end{align}
By taking the Frobenius norm, and carefully estimating the norm of the error term $ \big\| \sum_{j = 1}^{\tau - 1} E_{t + j - 1} J^{\tau - j}  \big\|_F^2$, and using Lemma~\ref{lemma:key} we can derive a recursion for the consensus distance. 
\begin{lemma}[Consensus distance recursion]\label{lem:consensus} There exists absolute constants $B_1, B_2, B_3 > 0$ such that  for a stepsize $\gamma < \frac{c}{B_3 L \tau}$
	\begin{align}
   \E \norm{\Psi_{t + k}}_F^2 \leq \frac{7}{8} \E \norm{\Psi_t}_F^2 +  \frac{1}{128 \tau} \sum_{j = 0}^{k - 1} \norm{\Psi_{t + j}}_F^2 + \frac{B_1 \tau L \gamma^2 }{c^2} \sum_{j = 0}^{k - 1} n e_{t + j} + \frac{B_2 \tau \gamma^2}{c^2} n\sigma^2 . \label{eq:y}
	\end{align}
\end{lemma}%
This lemma allows to replace $p$ with $c$ in the final convergence rate. This is achieved by grouping same gradients in the sum $ \big\| \sum_{j = 1}^{k - 1} E_{t + j - 1} J^{k - j}  \big\|_F^2$ and estimating the norm with Lemma~\ref{lem:norm_estimate}. %

An additional technical difficulty comes when unrolling consensus recursion \eqref{eq:y}. As iteration matrix $J$ is not contractive, i.e. $\norm{J} > 1$, then $\|\Psi_{t + j} \|_F^2$ for $j < \tau$ can be larger than $\|\Psi_t\|_F^2$ (up to $\approx \frac{1}{p^2}$ times as $\norm{J^i}^2 \leq \cO\left(\frac{1}{p^2}\right)~\forall i$). We introduce an additional term in the recursion that is provably non-increasing
\vspace*{-0.2cm}
\begin{align*}
\Phi_{t + \tau} := \frac{1}{\tau}\sum_{j = 0}^{\tau - 1} \|\Psi_{t + j} \|_F^2.
\end{align*}
With this we unroll consensus recursion.
\begin{lemma}[Unrolling recursion]\label{lem:unroll_rec}
	For $\gamma < \frac{c}{\sqrt{7 B_1} L \tau} \leq \frac{1}{2L\tau}$ it holds,
	\begin{align}\label{eq:unrolled_rec}
	\E \norm{\Psi_{t}}_F^2  &\leq \left(1 - \frac{1}{64\tau}\right)^{ t} A_0  +   \frac{22 B_1 \tau L \gamma^2}{c^2}  \sum_{j = 0}^{t - 1} \left(1 - \frac{1}{64\tau}\right)^{t - j} n e_j + \frac{20 B_2  \tau \gamma^2 }{c^2}  n\sigma^2 
	\end{align}
	where $e_j = \E [f(\bar \xx^{(j)}) - f(\xx^\star)]$, $A_0 =16 \|\Delta X^{(0)}\|_F^2 + \frac{24 \gamma^2}{p^2} \|\Delta Y^{(0)}\|^2_F $.
\end{lemma}

It remains to combine~\eqref{eq:y} and~\eqref{eq:unrolled_rec} using technique from \cite{koloskova2020unified}. \hfill $\Box$

\paragraph{Proof sketch of Theorem~\ref{thm:consensus}. } 
Using the matrix notation introduced above, the iterations of GT on problem~\eqref{eq:consensus} can be written in a simple form:
\begin{align*}
\begin{pmatrix}
\Delta X^{(t + 1)}\\
\gamma \Delta Y^{(t + 1)}
\end{pmatrix}^\top = 
\begin{pmatrix}
	\Delta X^{(t)}\\
	\gamma \Delta Y^{(t)}
	\end{pmatrix}^\top \underbrace{\begin{pmatrix}
		\tilde W & \gamma \left(W - I\right) \\
		- \tilde W & (1 - \gamma)\tilde W
		\end{pmatrix}}_{J'}\,.
\end{align*}
Similar as above, also the matrix $J'$ is not a contraction operator, but in contrast to $J$ it is diagonalizable: $J'=  Q \Lambda Q^{-1}$ for some $Q$ and diagonal $\Lambda$. It follows that $ \norm{(J')^t}^2 = \norm{Q \Lambda^t Q^{-1}}^2$ is decreasing as $(1 - p)^t \norm{Q}^2\norm{Q^{-1}}^2$. 
With this observation, the proof simplifies. \hfill $\Box$

\section{Experiments}
\label{sec:experiments}

In this section we investigate the tightness of parameters $c$ and $p$ in our theoretical result.

\textbf{Setup.} 
We consider simple quadratic functions defined as $f_i(\xx) = \norm{\xx}^2$, and $\xx^{(0)}$ is randomly initialized from a normal distribution $\cN(0, 1)$. We add artificially stochastic noise to gradients as $\nabla F_i(\xx, \xi) = \nabla f_i(\xx) + \xi$, where $\xi \sim \cN(0, \frac{\sigma^2}{d} I)$ so that Assumption~\ref{a:opt_nc} is satisfied.
We elaborate the details as well as results under other problem setups in Appendix~\ref{sec:exp_details}.

We verify the dependence on graph parameters $p$ and $c$ for the stochastic noise term. We fix the stepsize $\gamma$ to be constant, vary $p$ and $c$ and measure the value of $f(\bar \xx^{(t)} ) - f^\star$ that GT reaches after a large number of steps. According to the theory, GT converges to the level $\cO\left( \frac{\gamma \sigma^2}{n}  + \frac{\gamma^2 \sigma^2}{pc^2}\right)$ in a linear number of steps (to reach higher accuracy, smaller stepsizes must be used). To decouple the second term we need to ensure that the first term is small enough. For that, we take the number of nodes $n$ to be large. In all experiments we ensure that the first term is at least by order of magnitude smaller than the second by comparing the noise level with GT on a fully-connected topology.

\textbf{The effect of $p$.} First, in Figure~\ref{fig:p} we verify the expected $\cO\big(\frac{1}{p}\big)$ dependence when $c$ is a constant. For a fixed $n = 300$ number of nodes with $d = 100$ we vary the value of a parameter $p$ by interpolating the ring topology (with uniform weights) with the fully-connected graph.  The loss value $f(\xx^{(\infty)})$ scales linearly in $\frac{1}{p}$ as can be observed in Figure~\ref{fig:p} and the dependency on $p$ can thus not further be improved.

\begin{figure}[!h]
	\vspace{-1em}
	\centering
	\subfigure[\small $1/p$ (constant $c$).]{
		\includegraphics[width=.315\textwidth,]{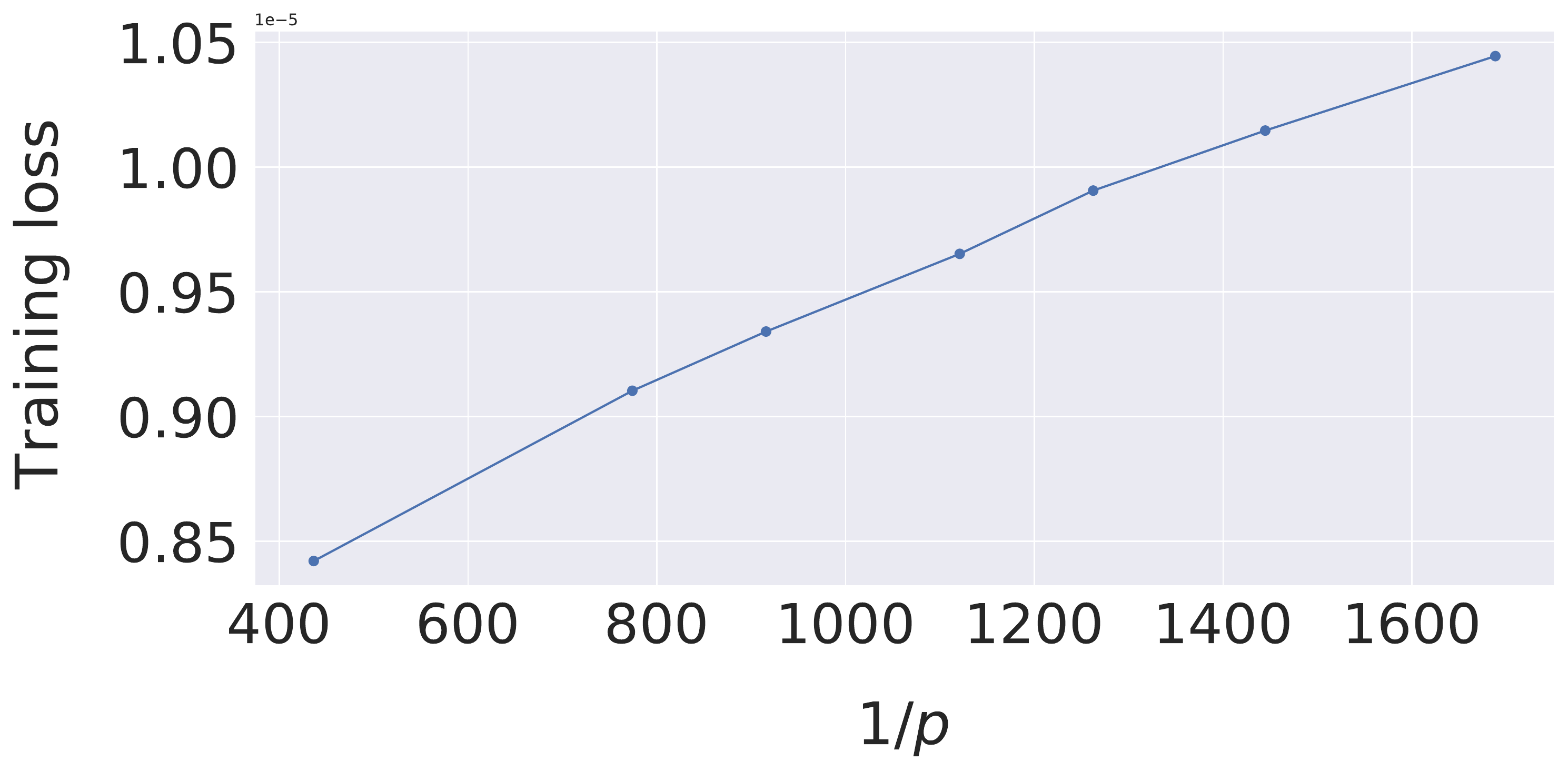}
		\label{fig:}
	}
	\hfill
	\subfigure[\small $1/p^2$ (constant $c$).]{
		\includegraphics[width=.315\textwidth,]{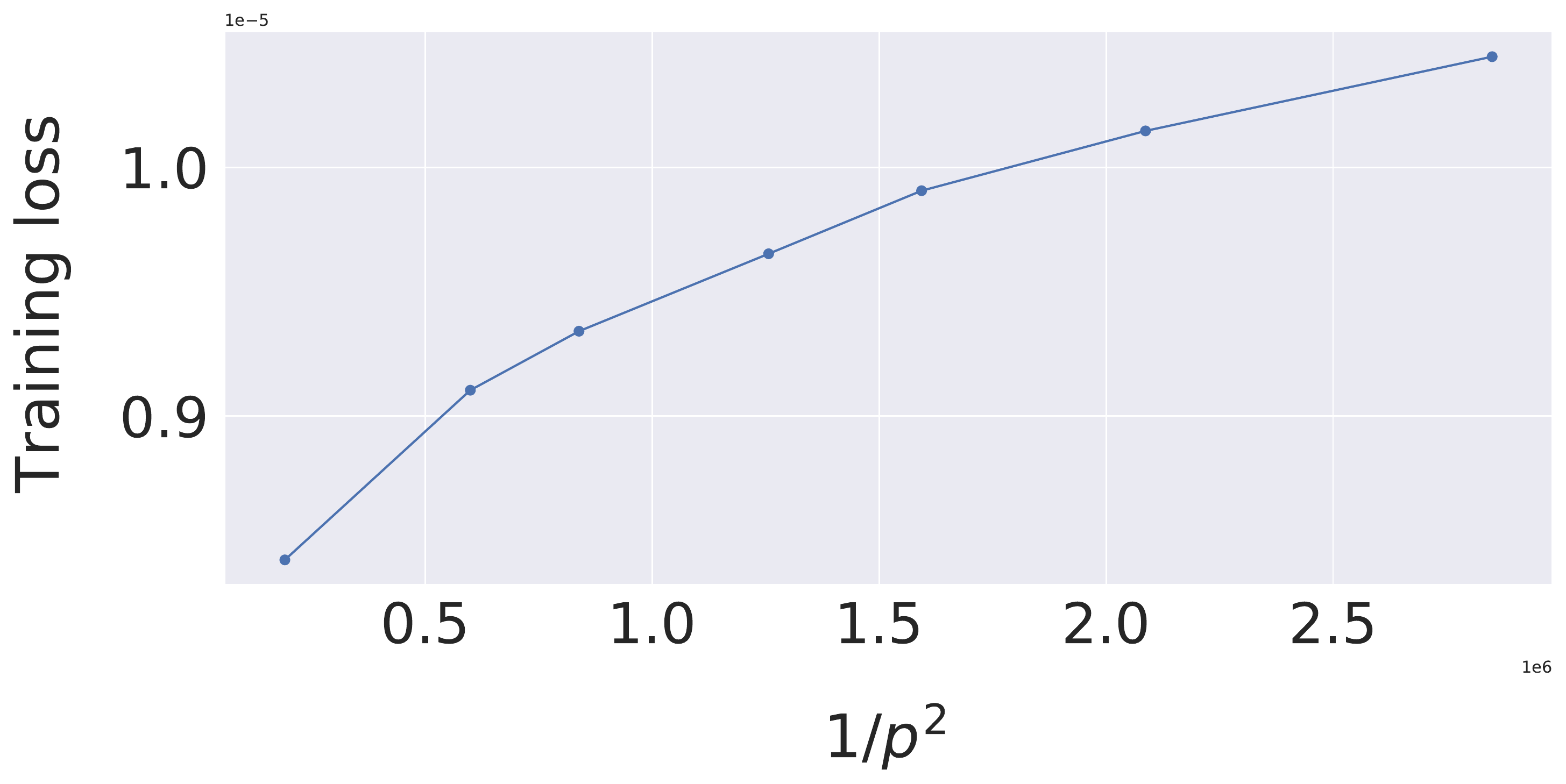}
		\label{fig:}
	}
	\hfill
	\subfigure[\small $1/p^3$ (constant $c$).]{
		\includegraphics[width=.315\textwidth,]{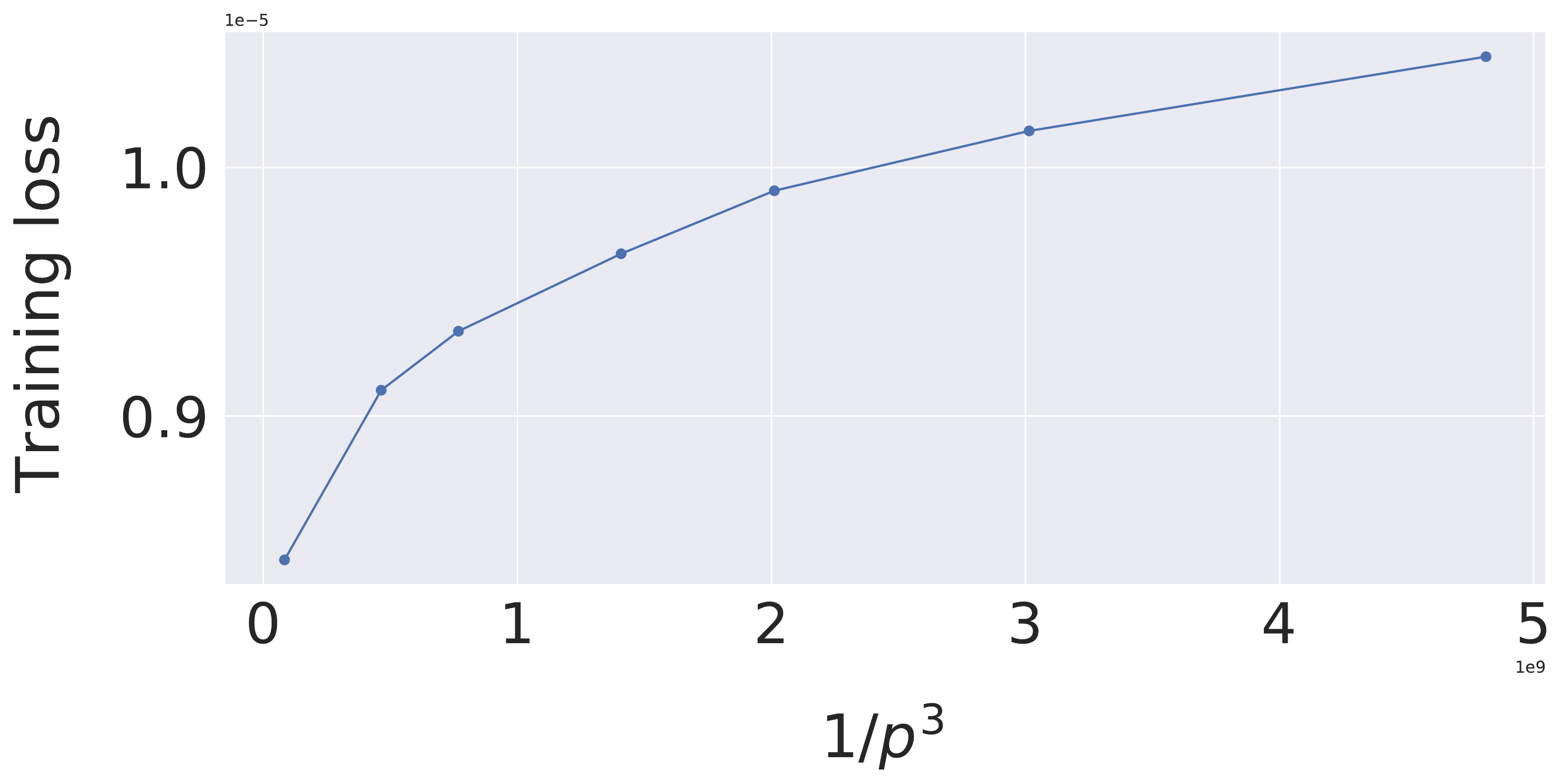}
		\label{fig:}
	}
	\caption{\small
		Impact of $p$ on convergence with the stochastic noise $\sigma^2 = 1$, when $c$ and $\gamma$ are kept constant.
		We see a linear scaling in $\frac{1}{p}$ that verifies the $\cO\big( \frac{1}{p}\big)$, dependence rather than prior predicted $\cO\big( \frac{1}{p^2}\big)$.
	}
	\label{fig:p}
\end{figure}

\textbf{The effect of $c$.} In Figure~\ref{fig:c} we aim to examine the dependence of the term $\cO\big(\frac{1}{pc^2} \big)$ on the parameter $c$, in terms of $1 / (pc^2)$ and $1 / (cp)$.
We take the ring topology on a fixed number of $n=300$ nodes and reduce the self-weights to achieve different values of $c$ (see appendix for details). Otherwise the setup is as above.
The current numerical results may suggest the existence of a potentially better theoretical dependence of the term $c$ (as discussed in Section~\ref{sec:consensus}); we leave the study for future work.

\begin{figure}[!h]
	\vspace{-1em}
	\centering
	\subfigure[\small $1/(pc^2)$ (constant $p$).]{
		\includegraphics[width=.45\textwidth,]{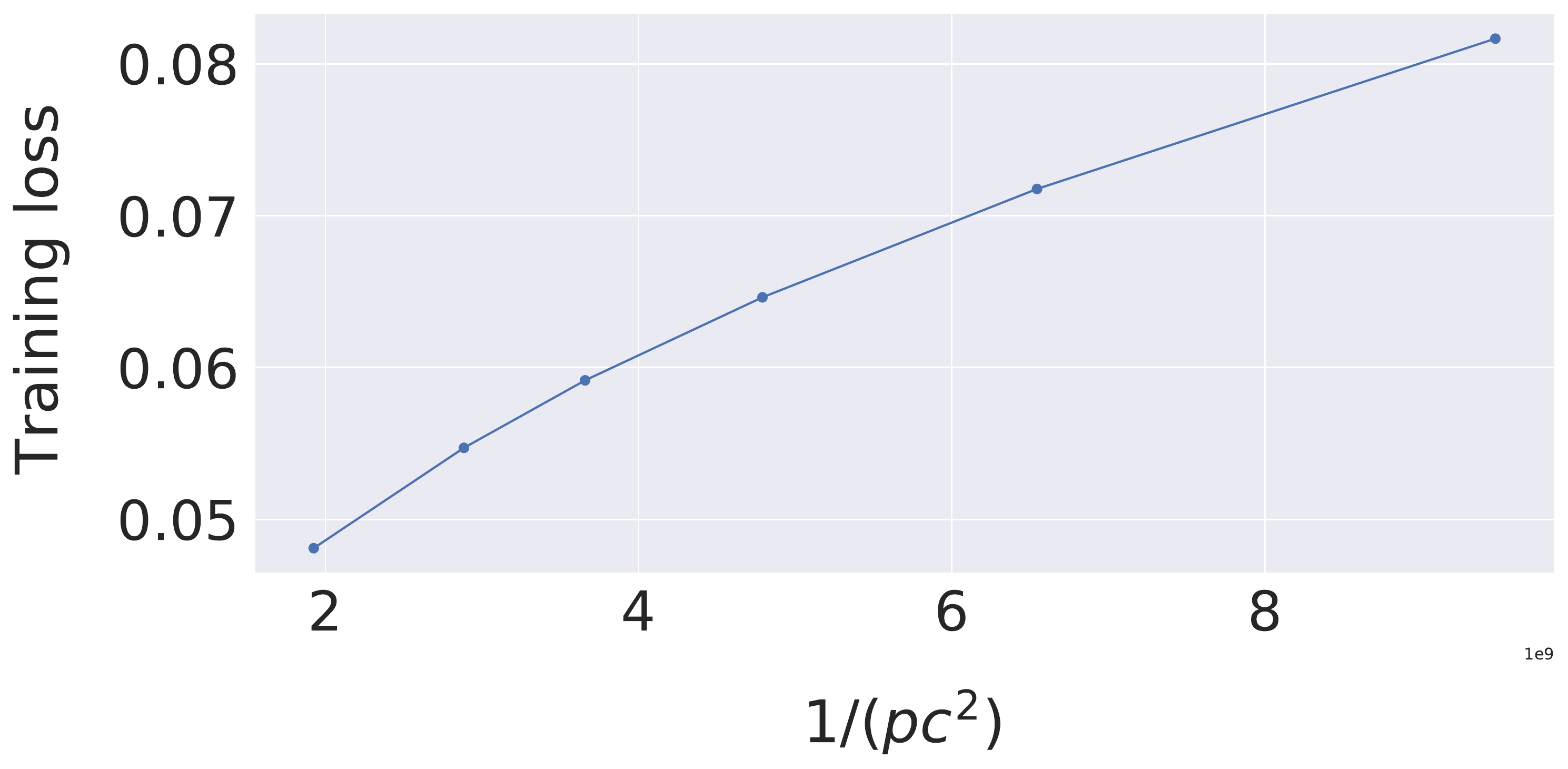}
		\label{fig:quadractics2_constant_p_inverse_c2p}
	}
	\subfigure[\small $1/(cp)$ (constant $p$).]{
		\includegraphics[width=.45\textwidth,]{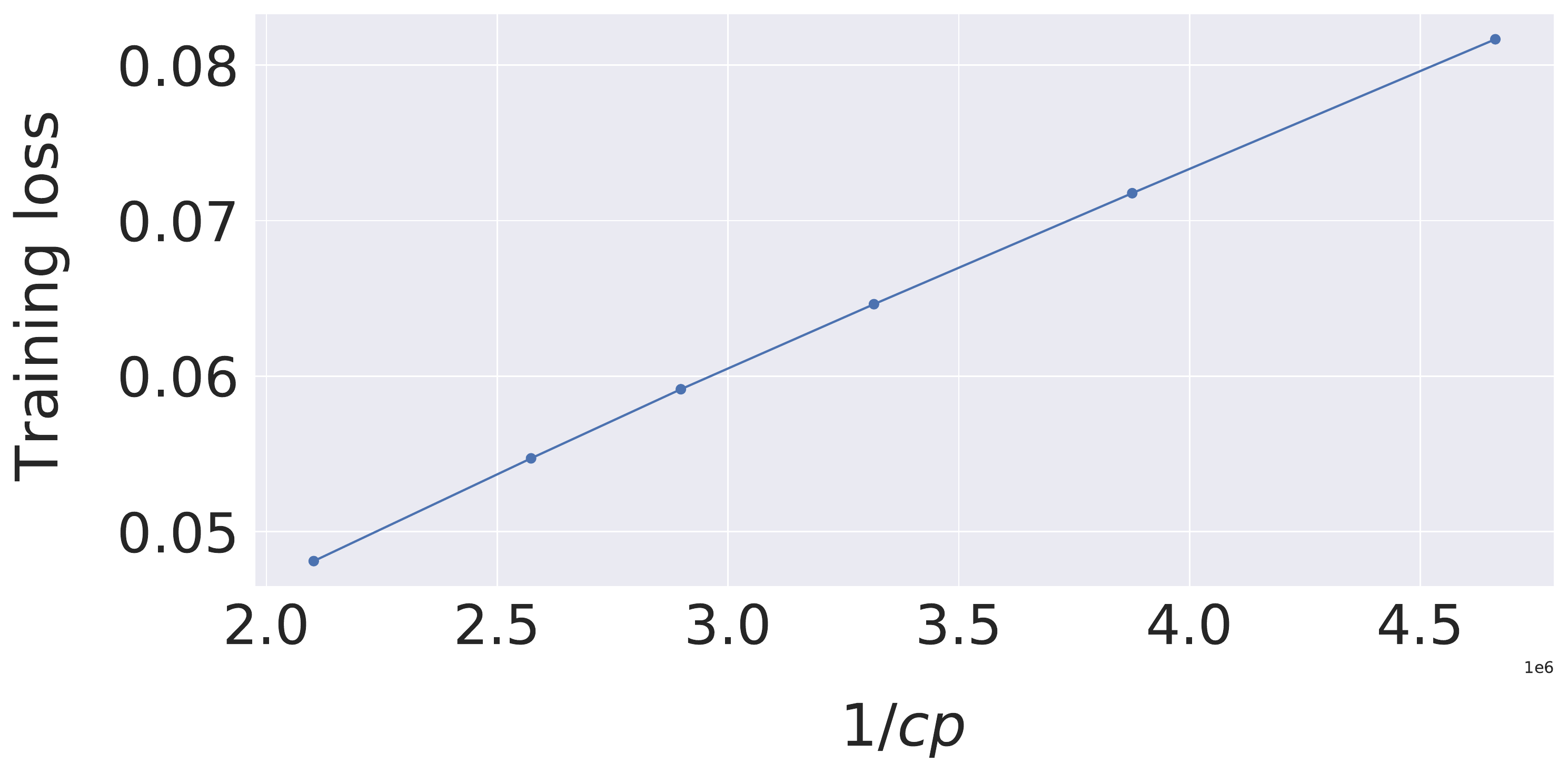}
		\label{fig:quadractics2_constant_p_inverse_cp}
	}

	\caption{\small
		Impact of $c$ on the convergence with the stochastic noise $\sigma^2 = 1$, when $p$ and $\gamma$ are kept constant.
		We see a near linear scaling in $\cO\bigl(\frac{1}{p c}\big)$ while the estimate $\cO \bigl(\frac{1}{p c^2}\bigr)$ appears to be too conservative on this problem.
	}
	\label{fig:c}
\end{figure}

\section{Conclusion}
We have derived improved complexity bounds for the GT method, that improve over all previous results. We verify the tightness of the second term in the convergence rate in numerical experiments. Our analysis identifies that the smallest eigenvalue of the mixing matrix has a strong impact on the performance of GT, however the smallest eigenvalue can often be controlled in practice by choosing large enough self-weights ($w_{ii}$) on the nodes.

Our proof technique might be of independent interest in the community and might lead to improved analyses for other gossip based methods where the mixing matrix is not contracting (for e.g.\ in directed graphs, or using row- or column-stochastic matrices).

\begin{ack}
This project was supported by SNSF grant 200020\_200342, EU project DIGIPREDICT, and a Google PhD Fellowship. The authors thank Martin Jaggi for his support. %
\end{ack}

\medskip

\bibliographystyle{plainnat-fixed}
{\small
\bibliography{reference}
}


\newpage
\appendix

\section{Proof of Theorem~\ref{thm:consensus} --- Consensus Functions}

We consider functions $f_i(\xx) = \frac{1}{2} \norm{\xx- \boldsymbol{\mu}_i}^2$, where $\xx, \boldsymbol{\mu}_i \in \R^d$. 
 Then $\nabla f_i(\xx) = \xx - \boldsymbol{\mu}_i$. In matrix notation, the GT algorithm in this special case is equivalent to

\begin{align*}
\begin{pmatrix}
X^{(t + 1)}\\
\gamma Y^{(t + 1)}
\end{pmatrix}^\top = \begin{pmatrix}
X^{(t)}\\
\gamma Y^{(t)}
\end{pmatrix}^\top \begin{pmatrix}
W & - W\\
0 & W 
\end{pmatrix}  + \gamma \begin{pmatrix}
0\\
X^{(t + 1)} - X^{(t)}
\end{pmatrix}^\top  =  \begin{pmatrix}
X^{(t)}\\
\gamma Y^{(t)}
\end{pmatrix}^\top\begin{pmatrix}
W & - W\\
\gamma (W - I) & (1 - \gamma)W 
\end{pmatrix} \,.
\end{align*}

The optimal point $\xx^\star = \boldsymbol{\bar\mu} = \frac{1}{n} \sum_{i = 1}^n \boldsymbol{\mu_i}$. Denote $X^\star = \left[\xx^\star, \dots , \xx^\star \right] \in \R^{d\times n}$.  We decompose the error as 
\begin{align*}
\norm{X^{(t)} - X^\star}_F^2 = \underbrace{\norm{X^{(t)} - \bar X^{(t)}}_F^2}_{\text{consensus error}} + \underbrace{\norm{\bar X^{(t)} - X^\star}^2}_{\text{optimization error}} \,.
\end{align*}
\paragraph{For the \textbf{optimization part},} notice that $\bar Y^{(t)} = \bar X^{(t)} - X^\star$. That is because 
\begin{align*}
\bar Y^{(0)} = \nabla f(X^{(0)}) \frac{1}{n}\1\1^\top = \bar X^{(0)} - X^\star, && \bar Y^{(t + 1)} = \bar Y^{(t)} + \bar X^{(t + 1)} - \bar X^{(t)} \,.
\end{align*}
Therefore, the optimization error is equal to 
\begin{align*}
\norm{\bar X^{t} - X^\star}_F^2 &= \norm{\bar X^{(t - 1)} - \gamma \bar Y^{(t - 1)} - X^\star}_F^2 = \norm{(1 - \gamma) \left( \bar X^{(t - 1)}  - X^\star\right)}_F^2 \\ &= (1 - \gamma)^{2t} \norm{\bar X^{(0)} - X^\star}_F^2\,.
\end{align*}
\paragraph{For the \textbf{consensus part},} denoting, $\tilde W = W - \frac{\1\1^\top}{n}$, $\Delta X^{(t)} = X^{(t)} - \bar X^{(t)}$, $\Delta Y^{(t)} = Y^{(t)} - \bar Y^{(t)}$, 
\begin{align*}
\begin{pmatrix}
\Delta X^{(t)}\\
\gamma \Delta Y^{(t)}
\end{pmatrix}^\top = 
\begin{pmatrix}
\Delta X^{(0)}\\
\gamma \Delta Y^{(0)}
\end{pmatrix}^\top \underbrace{\begin{pmatrix}
\tilde{W} & - \tilde{W}\\
\gamma(W - I) & (1 - \gamma)\tilde{W}
\end{pmatrix}^t}_{J'} \,.
\end{align*}
Taking the norm,
\begin{align*}
\norm{\Delta X^{(t)}}_F^2 + \gamma^2 \norm{\Delta Y^{(t)}}_F^2 \leq \norm{J'^t}_2^2 \left( \norm{\Delta X^{(0)}}_F^2 + \gamma^2 \norm{\Delta Y^{(0)}}_F^2\right) \,.
\end{align*}

Lets analyze spectral properties of matrix $J'^t$. Let the eigenvalue decomposition of $W$ be $W = U \Lambda U^\top$, the eigenvalue decomposition of $\tilde W$ is $\tilde{W} = U \tilde \Lambda U^\top$ for diagonal $\tilde \Lambda$.

We can decompose 
\begin{align*}
J' = \begin{pmatrix}
U &  0 \\
0 & U
\end{pmatrix}
\underbrace{\begin{pmatrix}
\tilde \Lambda & - \tilde \Lambda\\
\gamma\left(\Lambda - I \right)& (1 - \gamma)\tilde \Lambda
\end{pmatrix}}_{=: M}\begin{pmatrix}
U^\top &  0 \\
0 & U^\top
\end{pmatrix} \,.
\end{align*}
And, 
\begin{align*}
\norm{J'^t}_2^2 = \norm{\begin{pmatrix}
U &  0 \\
0 & U
\end{pmatrix}
\begin{pmatrix}
\tilde \Lambda & - \tilde \Lambda\\
\gamma\left( \Lambda - I\right) & (1 - \gamma)\tilde \Lambda
\end{pmatrix}^t\begin{pmatrix}
U^\top &  0 \\
0 & U^\top
\end{pmatrix}}_2^2 = \norm{\begin{pmatrix}
\tilde \Lambda & - \tilde \Lambda\\
\gamma\left( \Lambda - I\right) & (1 - \gamma)\tilde \Lambda
\end{pmatrix}^t}_2^2,
\end{align*}
where the last equality is due to unitary property of $U$. 

\begin{lemma}
	To diagonalize a block-diagonal matrix 
	\begin{align*}
	\begin{pmatrix}
	A &  B \\
	C & D
	\end{pmatrix},
	\end{align*}
	where $A = \diag(a_0,\dots a_n) \in R^{n\times n}$, $B = \diag(b_0, \dots, b_n)$, $C = \diag(c_0, \dots, c_n)$, $D = \diag(d_0, \dots, d_n)$. 
	Assume that each of the $2\times 2$ matrices 
	\begin{align*}
	\begin{pmatrix}
	a_i & b_i \\
	c_i & d_i
	\end{pmatrix}
	\end{align*}
	are diagonalizable with 
		\begin{align*}
	\begin{pmatrix}
	a_i & b_i \\
	c_i & d_i
	\end{pmatrix} = \begin{pmatrix}
	q^{(1)}_i & q^{(2)}_i \\
	q^{(3)}_i & q^{(4)}_i
	\end{pmatrix} \cdot \begin{pmatrix}
	d^{(1)}_i &  0 \\
	0 & d^{(2)}_i
	\end{pmatrix} \cdot \begin{pmatrix}
	q^{(-1)}_i & q^{(-2)}_i \\
	q^{(-3)}_i & q^{(-4)}_i
	\end{pmatrix}
	\end{align*}
	Then the original matrix is diagonalizable and its diagonalization is equal to 
	\begin{align*}
		\begin{pmatrix}
	A &  B \\
	C & D
	\end{pmatrix} = \begin{pmatrix}
	Q_1 & Q_2\\
	Q_3 & Q_4
	\end{pmatrix}\cdot \begin{pmatrix}
	D_1 & 0\\
	0 & D_2
	\end{pmatrix}\cdot \begin{pmatrix}
	Q_{-1} & Q_{-2}\\
	Q_{-3} & Q_{-4}
	\end{pmatrix},
	\end{align*}
	where each $Q_l = \diag\left( q_1^{(l)}, \dots, q_n^{(l)} \right)$, $D_l = \diag\left(d_1^{(l)}, \dots d_n^{(l)}\right)$. 
	
\end{lemma}

We need to show that the following $2\times2 $ matrices are diagonalizable.
\begin{align*}
M_i := \begin{pmatrix}
\lambda_i & -\lambda_i \\
\gamma \left( \lambda_i - 1\right) & (1 - \gamma) \lambda_i
\end{pmatrix},
\end{align*} 
where the $\lambda_i$ are eigenvalues of the matrix $\tilde{W}$.
The eigenvalues of $M_i$ are
\begin{align*}
 \lambda(M_i) = \left\{\lambda_i- \frac{\gamma \lambda_i}{2} - \frac{1}{2} \sqrt{\gamma \lambda_i} \sqrt{4 + (\gamma - 4) \lambda_i}, \lambda_i - \frac{\gamma \lambda_i}{2} + \frac{1}{2} \sqrt{\gamma \lambda_i} \sqrt{4 + (\gamma - 4) \lambda_i} \right\}\,,
\end{align*}
which are distinct for $\gamma > 0$, therefore the matrix is diagonalizble (over $\mathbb{C}$).

If $\lambda_i$ is positive, then by choosing $\gamma \leq 1-\lambda_i$,
\begin{align*}
\abs{\lambda(M_i)} \leq  \frac{1}{3}\lambda_i + \frac{2}{3} \,.
\end{align*}
If $\lambda_i$ is negative, then, then by choosing $\gamma \leq 1-\abs{\lambda_i}$,
\begin{align*}
 \abs{\lambda(M_i)} \leq \frac{1}{3} \abs{\lambda_i} + \frac{2}{3} \,.
\end{align*} 
We do not give the full formal prove of these two bounds. First we note that $\abs{(M_i)}$ is monotone in $\gamma$, i.e.\ the absolute value increases in $\gamma$. Therefore it is enough to check that it holds $\abs{\lambda(M_i)} \leq \frac{1}{3} \abs{\lambda_i} + \frac{2}{3}$ for $\gamma = 1- \abs{\lambda_i}$. We visualize these upper bounds with Mathematica~\cite{Mathematica} in Figure~\ref{fig:visualproof}.
\begin{figure}[H] \centering
\includegraphics[scale=0.5]{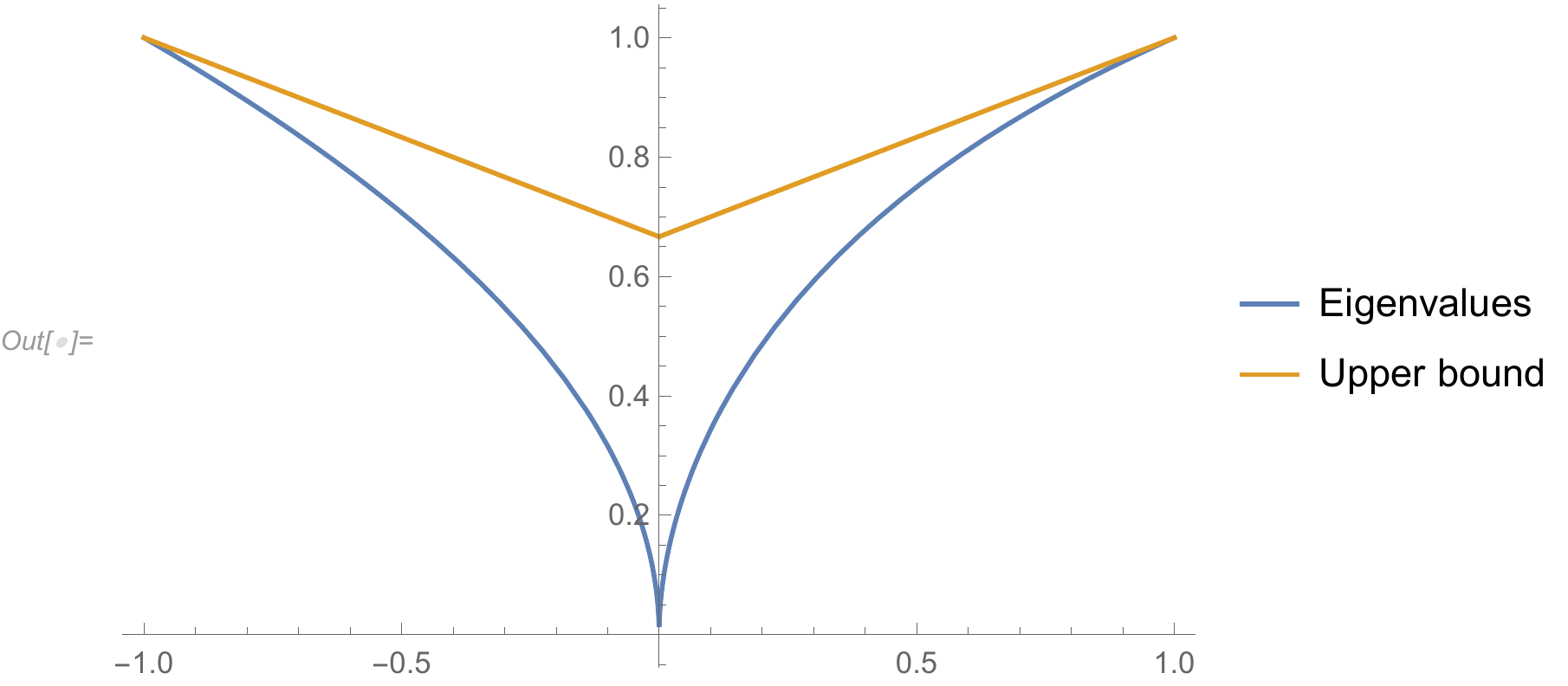}
\caption{The upper bound $\frac{1}{3}\abs{\lambda_i} + \frac{2}{3}$ (yellow) vs.\ the true $\abs{\lambda(M_i)}$ for the choice $\gamma = 1-\abs{\lambda_i}$.}
\label{fig:visualproof}
\end{figure}
This concludes the proof.

\section{Proof of Theorem~\ref{thm:GT-better-upper-bound} --- General Case}

We first re-state theorem~\ref{thm:GT-better-upper-bound} in terms of number of iterations $T$
\begin{theorem}\label{thm:GT-upper-bound-T}
	For GT algorithm~\ref{alg:gt} with a mixing matrix as in Definition~\ref{def:W}, under Assumptions~\ref{a:W}, \ref{a:lsmooth_nc}, \ref{a:opt_nc}, after $T$ iterations, if $T > \frac{2}{p}\log\left(\frac{50}{p} (1+ \log \frac{1}{p}) \right)$, there exists a constant stepsize $\gamma_t = \gamma$ such that the error is bounded as 
	\\
	\textbf{Non-convex:} 
	\begin{align*}
	\frac{1}{T + 1} \sum_{t = 0}^T \norm{\nabla f(\bar\xx^{(t)})}_2^2 \leq \tilde\cO\left(\sqrt{\frac{LF_0 \sigma^2}{nT}} +  \left(\frac{\sigma L F_0}{(\sqrt{p} c + p\sqrt{n}  )T} \right)^{\nicefrac{2}{3}} + \frac{L (F_0 +L \tilde R_0^2) }{pc T}\right)\,,
	\end{align*}\\
	\textbf{Strongly-convex:}
	Under additional Assumption \ref{a:mu-convex} with $\mu > 0$, it holds
	\begin{align*}
	\sum_{t = 0}^T \frac{w_t}{W_T} \left[\E f(\bar\xx^{(t)}) - f^\star \right] + \frac{\mu}{2} R_{T + 1}\leq  \tilde\cO\left(\frac{\sigma^2}{\mu nT} +  \frac{L  \sigma^2}{\mu^2 p c^2 T^2}+ \frac{L (R_0^2 + \frac{L}{\mu}\tilde R_0^2)}{p c} \exp\left[ -\frac{\mu p c T}{L} \right]\right)\,,
	\end{align*}
	\\
	\textbf{Weakly-convex:}
	Under Assumptions \ref{a:mu-convex} with $\mu \geq 0$, it holds
	\begin{align*}
	\frac{1}{T + 1}\sum_{t = 0}^T \left[\E f(\bar\xx^{(t)}) - f^\star \right] \leq  \tilde\cO\left(\sqrt{\frac{R_0^2 \sigma^2}{nT}} +  \left(\frac{\sigma \sqrt{L} R_0^2 }{\sqrt{p} c T} \right)^{\nicefrac{2}{3}} + \frac{L (R_0^2 + \tilde R_0^2) }{pc T} \right)\,,
	\end{align*}
	where $F_0 = f(\bar\xx^{(0)}) - f^\star$, $R_t = \norm{\xx^{(t)} - \xx^\star}$, $t \in \{0, T + 1\}$,  $\tilde R_0^2 = \frac{1}{n} \sum_{i = 1}^n \norm{\xx_i^{(0)} - \bar \xx^{(0)}}^2 + \frac{1}{n} \sum_{i = 1}^n \norm{\yy_i^{(0)} - \bar \yy^{(0)}}^2$.
\end{theorem}

\subsection{Useful Inequalities}
\begin{proof}[Proof of Lemma~\ref{lemma:key}]
By monotonicity, it suffices to check the inequality for $i=\tau$. By using $(1-p)^i \leq e^{-ip}$ and plugging $\tau$ into~\eqref{eq:i-p} it follows:
\begin{align*}
\norm{J^i}^2 &\leq 
e^{-\tau p} (1+\tau^2) \leq \frac{p^2}{50^2 (1+\log\frac{1}{p})^2}\left(1 + \frac{(2(\log(50) + \log(\frac{1}{p}(1+\log\frac{1}{p})))^2}{p^2} \right) \\ &\leq \frac{1}{50^2} + \frac{1}{10} + \frac{1}{4}
\end{align*}
with $\log(\frac{1}{p}(1+ \log\frac{1}{p}) \leq \log\frac{1}{p}+\log\log\frac{1}{p} \leq 2\log\frac{1}{p}$, then $(\log(4)+2\log\frac{1}{p})^2\leq 2\log 4 + 8 \log \frac{1}{p}$, and
$(4\log 50 + 16 \log \frac{1}{p}))^2 \leq (128 + 512 \log \frac{1}{p})$.
\end{proof}

\begin{lemma}
\label{lemma:ilambdai}
Let $\lambda \in (-1,1)$ with $\abs{\lambda} = 1-\alpha$, for $0 < \alpha < 1$. Then $\abs{i \lambda^i} \leq \frac{1}{\alpha}$ for all $i \geq 0$.
\end{lemma}
\begin{proof}
\hfill $\displaystyle
 \abs{i \lambda^i} 
 \leq i (1-\alpha)^i  \leq \sum_{j=1}^i (1-\alpha)^j  \leq \frac{1-\alpha}{\alpha}  \,.$ \hfill \null
\end{proof}

\begin{lemma}[fact]
\label{fact:1}
Let $W$ be a symmetric matrix with eigenvalues $\lambda_1(W)\geq \dots \lambda_n(W)$. Then $\norm{W}^2 = \max_i \lambda_i^2(W)$.
\end{lemma}

\begin{lemma}\label{lem:norm_estimate}
It holds
$\norm{(i+1)\tilde W^{i+1} - i \tilde W^i}^2 \leq \frac{4}{\alpha^2} \leq \frac{16}{c^2}$ for all $i \geq 0$, where $\alpha = 1-|\lambda_n(W)|$ and $c$ as defined in~\eqref{def:p}.
\end{lemma}
\begin{proof}
The eigenvalues of $(i+1)\tilde { W}^{i+1} - i \tilde W^i$ have the form $(i+1) \lambda^{i+1} - i \lambda^i$, for $\lambda \in \Lambda:=\{\lambda_1(\tilde{W}),\dots, \lambda_n(\tilde{W})\}$, the eigenvalues of $\tilde W$. By Lemma~\ref{fact:1}, it holds
\begin{align*}
 \norm{(i+1)\tilde{W}^{i+1} - i \tilde{ W}^i}^2 = \max_{\lambda \in \Lambda} ((i+1) \lambda^{i+1} - i \lambda^i)^2 \,.
\end{align*}
If the maximum is attained for a positive $\lambda > 0$, we conclude
\begin{align*}
 ((i+1) \lambda^{i+1} - i \lambda^i)^2 
 &= (\lambda^{i+1} - i \lambda^i(1-\lambda))^2 \\
 &\leq 2 (\lambda^{i+1})^2 + 2(1-\lambda)^2 (i \lambda^i)^2  \\
 &\leq 2 (\lambda^{i+1})^2 + 2 \frac{(1-\lambda)^2}{(1-\lambda)^2} \\
 &\leq 4
\end{align*}
with Lemma~\ref{lemma:ilambdai} for the first estimate and using $\lambda \leq 1$ on the last line. If the maximum is attained for a negative $\lambda < 0$ with $\lambda = -1 + \beta$, for $\beta > 0$, then 
\begin{align*}
 ((i+1) \lambda^{i+1} - i \lambda^i)^2 
 &\leq 2 ((i+1) \lambda^{i+1})^2 + 2(i \lambda^i)^2 \\
 &\leq  \frac{2}{\beta^2} +  \frac{2}{\beta^2} \leq \frac{4}{\alpha^2}
\end{align*}
with Lemma~\ref{lemma:ilambdai} and $\alpha \leq \beta$.

Note that $c = 1 - (1-\alpha)^2 = 2 \alpha - \alpha^2 \geq \alpha$, since $\alpha (1 - \alpha) \geq 0$ and that $c \leq 2 \alpha$.
\end{proof}

\begin{lemma}\label{lem:norm_estimate2}
It holds $\norm{i \tilde{W}^i}^2 \leq \frac{1}{\alpha^2} \leq \frac{4}{p^2}$. 
\end{lemma}
\begin{proof}
$\norm{i \tilde W^i}^2 = \left(i \norm{\tilde W^i}\right)^2$, and the proof follows with Lemma~\ref{lemma:ilambdai} and~\ref{fact:1} from above.
\end{proof}

\begin{lemma}\label{lem:norm_J}
	It holds $\norm{\Psi^{0} J^t }^2_F \leq 2 \norm{\Delta X^{(0)}}^2_F + \frac{3 \gamma^2}{p^2}\norm{\Delta Y^{(0)}}^2_F$ for all $t \geq 0$, where $p$ is defined in \eqref{def:p}.
\end{lemma}
\begin{proof}
	Starting from \eqref{eq:i-p} and using Lemma~\ref{lemma:ilambdai} with $\delta = 1 - \lambda_2$
	\begin{align*}
	\norm{\Psi^{0}  J^i }^2_F  = \norm{\begin{pmatrix}
	\Delta X^{(0)} \tilde W^i  - i \gamma \Delta Y^{(0)}\tilde W^i \\
	\gamma \Delta Y^{(0)}\tilde W^i
	\end{pmatrix}^\top }^2_F \leq 2 \norm{\Delta X^{(0)}}^2_F + \frac{3 \gamma^2}{p^2}\norm{\Delta Y^{(0)}}^2_F.
	\end{align*}
\end{proof}

\begin{lemma}\label{remark:norm_of_sum}
	For arbitrary set of $n$ vectors $\{\aa_i\}_{i = 1}^n$, $\aa_i \in \R^d$
	\begin{equation}\label{eq:norm_of_sum}
	\norm{\sum_{i = 1}^n \aa_i}^2 \leq n \sum_{i = 1}^n \norm{\aa_i}^2 \,.
	\end{equation}
\end{lemma}
\begin{lemma}\label{remark:scal_product}
	For given two vectors $\aa, \bb \in \R^d$
	\begin{align}\label{eq:scal_product}
	&2\lin{\aa, \bb} \leq \gamma \norm{\aa}^2 + \gamma^{-1}\norm{\bb}^2\,, & &\forall \gamma > 0 \,.
	\end{align}
\end{lemma}
\begin{lemma}\label{remark:norm_of_sum_of_two}
	For given two vectors $\aa, \bb \in \R^d$ %
	\begin{align}\label{eq:norm_of_sum_of_two}
	\norm{\aa + \bb}^2 \leq (1 + \alpha)\norm{\aa}^2 + (1 + \alpha^{-1})\norm{\bb}^2,\,\, & &\forall \alpha > 0\,.
	\end{align}
	This inequality also holds for the sum of two matrices $A,B \in \R^{n \times d}$ in Frobenius norm.
\end{lemma}

\begin{lemma}\label{rem:frobenious_norm_of_matrix_mult}
	For $A\in \R^{d\times n}$, $B\in \R^{n\times n}$
	\begin{align}\label{eq:frob_norm_of_multiplication}
	\norm{AB}_F \leq \norm{A}_F \norm{B}_2 \,.
	\end{align}
\end{lemma}

\subsection{Convex Cases}
\paragraph{Proof of Lemma~\ref{lem:consensus}}
We first state auxiliary lemma about consensus recursion.
\begin{lemma}\label{lem-aux}
	There exists absolute constants $C_1 = 440, C_2 = 380$ such that iterates of Algorithm~\ref{alg:gt} satisfy,
	\begin{align}
	\E \norm{\Psi_{t + k}}_F^2 \leq \frac{3}{4} \E \norm{\Psi_t}_F^2 +  \gamma^2 \frac{C_1 \tau}{c^2} \sum_{j = 0}^{k - 1} \E \norm{\nabla f(X^{t + j}) - \nabla f(X^\star)}_F^2 + \gamma^2 \frac{C_2 \tau}{c^2} n \sigma^2 \,. \label{eq:y_1}
	\end{align}
	where $\tau \leq k \leq  2 \tau$, $\tau = \frac{2}{p}\log\left(\frac{50}{p} (1+ \log \frac{1}{p})\right) + 1$, $p$ and $c$ are defined in \eqref{def:p}, $\Psi_t = \left( \Delta X^{(t)}, \gamma \Delta Y^{(t)}\right)$ and is defined in \eqref{eq:25}.
\end{lemma}
\begin{proof}
	We start from the recursion \eqref{eq:mid-proof} given in the main text 
	\begin{align*}
	\Psi_{t + k} = \Psi_t J^k + \gamma \sum_{j = 1}^{k}  E_{t + j - 1}J^{k - j}\,.
	\end{align*}
	Taking the norm,
	\begin{align*}
	\norm{\Psi_{t + k}}_F^2 \stackrel{\eqref{eq:norm_of_sum_of_two}, \alpha = \frac{1}{4}, \eqref{eq:frob_norm_of_multiplication}}{\leq} \left(1 + \frac{1}{4}\right)\norm{J^k}_2^2 \norm{\Psi_t }_F^2  + 5 \gamma^2 \norm{\sum_{j = 1}^{k}  E_{t + j - 1}J^{k - j}}_F^2
	\end{align*}
	Using the key Lemma~\ref{lemma:key}, the first term can be estimated as 
	\begin{align*}
	\left(1 + \frac{1}{4}\right)\norm{J^k}_2^2\norm{\Psi_t }_F^2 \leq \frac{3}{4}\norm{\Psi_t }_F^2 \,.
	\end{align*}
	Lets estimate separately the second term. Denoting $G^{(t)} = \nabla F(X^{(t)}, \xi^{(t)})$, 
	\begin{align*}
	\norm{\sum_{j = 1}^{k }  E_{t + j - 1}J^{k - j}}_F^2 &= \norm{
		\begin{pmatrix}
		- \sum_{j = 1}^{k} \left(G^{(t + j)} - G^{(t + j - 1)}\right) (k - j) \tilde W^{k - j} (I - \frac{\1\1^\top}{n})\\
		\sum_{j = 1}^{k } \left(G^{(t + j)} - G^{(t + j - 1)}\right)\tilde W^{k - j} (I - \frac{\1\1^\top}{n})
		\end{pmatrix}}_F^2 \\
		&\stackrel{\eqref{eq:frob_norm_of_multiplication}}{\leq} \underbrace{\norm{\sum_{j = 1}^{k} \left(G^{(t + j)} - G^{(t + j - 1)}\right) (k - j) \tilde W^{k - j} }_F^2}_{=: T_1} \\ &\qquad\qquad + \underbrace{\norm{\sum_{j = 1}^{k } \left(G^{(t + j)} - G^{(t + j - 1)}\right)\tilde W^{k - j}}_F^2}_{=: T_2} \,,
	\end{align*}
	where we used the definition of the Frobenius norm and $\norm{ I - \frac{\1\1^\top}{n} } \leq 1$. We now give upper bounds for $T_1$ and $T_2$.

\paragraph{The second term $T_2$.} We firstly separate the stochastic noise by adding and subtracting the full gradient, 
\begin{align*}
T_2 &\stackrel{\eqref{eq:norm_of_sum_of_two}}{\leq} 3 \norm{\sum_{j = 1}^{k } \left(\nabla f(X^{(t +j)}) - \nabla f(X^{(t + j - 1)})\right)\tilde W^{k - j}}_F^2 \\
 &  \qquad + 3 \norm{\sum_{j = 1}^{k } \left(G^{(t + j)} - \nabla f(X^{(t +j)})\right)\tilde W^{k - j}}_F^2 + 3 \norm{\sum_{j = 1}^{k }  \left(G^{(t + j - 1)} - \nabla f(X^{(t + j - 1)})\right)\tilde W^{k - j} }_F^2\,.
\end{align*}
Note that 
\begin{align*}
 \E \norm{ \sum_{j=1}^k  \left(G^{(t + j)} - \nabla f(X^{(t +j)})\right)\tilde W^{k - j} }_F^2 &= \sum_{j=1}^k \E \norm{ \left(G^{(t + j)} - \nabla f(X^{(t +j)})\right)\tilde W^{k - j} }_F^2 \\
 &\stackrel{ \norm{\tilde W}\leq 1}{\leq}  \sum_{j=1}^k \E \norm{ G^{(t + j)} - \nabla f(X^{(t +j)}) }_F^2\,,
\end{align*}
where we used 
the martingale property $\E_j \left[G^{(j)} - \nabla f(X^{(j)}) \mid X^{(j)} \right] = 0$ for all $j \leq t$.
It follows
\begin{align*}
\E[T_2] &\stackrel{\eqref{eq:noise_opt_nc}}{\leq} 3 \E \norm{\sum_{j = 1}^{k } \left(\nabla f(X^{(t +j)}) - \nabla f(X^{(t + j - 1)})\right)\tilde W^{k - j}}_F^2 + 6k n \sigma^2\,.
\end{align*}
We expand further by adding and subtracting $\nabla f(X^\star)$ to the first norm, and bounding stochastic noise by \eqref{eq:noise_opt_nc} in the other terms
\begin{align*}
\E[T_2] &\stackrel{\eqref{eq:norm_of_sum_of_two}, \eqref{eq:noise_opt_nc}}{\leq} 6 \E \norm{\sum_{j = 1}^{k } (\nabla f(X^{(t +j)}) -\nabla f(X^\star)) \tilde W^{k - j}}_F^2 + 6 \E \norm{ \sum_{j = 1}^{k } (\nabla f(X^{(t + j - 1)}) - \nabla f(X^\star) ) \tilde W^{k - j}}_F^2 + 6k n \sigma^2\\
& \stackrel{\eqref{eq:norm_of_sum}, \eqref{eq:frob_norm_of_multiplication}}{\leq} 12 k \sum_{j = 0}^k \E \norm{\nabla f(X^{(t +j)}) -\nabla f(X^\star)}_F^2 + 6 k n \sigma^2\,.
\end{align*}
\textbf{The first term $T_1$.}
First, we separate the stochastic noise similarly as above. Defining $Z^{(t)}=G^{(t)}-\nabla f(X^{(t)})$, 
\begin{align*}
T_1 \stackrel{\eqref{eq:norm_of_sum_of_two}}{\leq} 2 \norm{\sum_{j = 1}^{k} \left[\nabla f ( X^{(t + j)}) - \nabla f (X^{(t + j - 1)})\right] (k - j) \tilde W^{k - j} }_F^2 + 2\norm{\sum_{j = 1}^{k} \left(Z^{(t + j)} - Z^{(t + j - 1)}\right) (k - j) \tilde W^{k - j} }_F^2. 
\end{align*}
Next, we add and subtract $\nabla f(X^\star)$ in the first term $k - 1$ times and temporarily denote $D^{(j)} = \nabla f(X^{(j)} ) - \nabla f(X^\star)$
\begin{align*}
T_1 & \leq  2 \norm{\sum_{j = 1}^{k} \left(D^{(t + j)} - D^{(t + j - 1)}\right)(k - j) \tilde W^{k - j} }_F^2 + 2\norm{\sum_{j = 1}^{k} \left(Z^{(t + j)} - Z^{(t + j - 1)}\right) (k - j) \tilde W^{k - j} }_F^2. 
\end{align*}
Next, we re-group the sums by the gradient index.
\begin{align*}
T_1 &\leq  2 \norm{ D^{(t + k - 1)} \tilde W  - (k - 1)  D^{(t)}\tilde W^{k - 1} + \sum_{j = 1}^{k - 2} D^{(t + j)}  \left[(k - j)\tilde W^{k - j} - (k - j - 1) \tilde W^{k - j - 1} \right]}_F^2\\
& \quad + 2 \norm{ Z^{(t + k - 1)} \tilde W  - (k - 1)  Z^{(t)}\tilde W^{k - 1} + \sum_{j = 1}^{k - 2} Z^{(t + j)}  \left[(k - j)\tilde W^{k - j} - (k - j - 1) \tilde W^{k - j - 1} \right]}_F^2 \\
& \stackrel{\eqref{eq:norm_of_sum}, \eqref{eq:frob_norm_of_multiplication}}{\leq} 2 k \left[\norm{D^{(t + k - 1)}}_F^2 + \norm{D^{(t)}(k - 1)\tilde W^{k - 1}}_F^2 + \sum_{j = 1}^{k - 2} \norm{D^{(t + j)}  \left[(k - j)\tilde W^{k - j} - (k - j - 1) \tilde W^{k - j - 1} \right]}_F^2  \right] \\
&\quad \quad \quad +2 \left[\norm{Z^{(t + k - 1)}}_F^2 + \norm{Z^{(t)}(k - 1)\tilde W^{k - 1}}_F^2 + \sum_{j = 1}^{k - 2} \norm{Z^{(t + j)}  \left[(k - j)\tilde W^{k - j} - (k - j - 1) \tilde W^{k - j - 1} \right]}_F^2  \right] 
\end{align*}
where for splitting $Z$ we used martingale property $\E_j \left[G^{(j)} - \nabla f(X^{(j)}) \mid X^{(j)} \right] = 0$ for all $j \leq t$. Next, we use Lemma~\ref{lem:norm_estimate} to estimate the norm $\norm{(k - j)\tilde W^{k - j} - (k - j - 1) \tilde W^{k - j - 1} }_2^2 \leq \frac{16}{c^2}$; and  using \eqref{eq:i-p} we estimate $\norm{(k - 1)\tilde W^{k - 1} }_2^2 \leq \norm{J^{k - 1}}^2 \leq \frac{1}{2}$ due to our choice of $k \geq \tau$ and a key Lemma~\ref{lemma:key}
\begin{align*}
T_1 &\stackrel{\eqref{eq:frob_norm_of_multiplication}}{\leq} \frac{32 k }{c^2} \sum_{j = 0}^{k - 1} \norm{D^{(t + j)}}_F^2  + \frac{32}{c^2}\sum_{j = 0}^{k - 1} \norm{Z^{(t + j)}}_F^2
\end{align*}
Taking expectation over the stochastic noise, 
\begin{align*}
\E [T_1] &\stackrel{\eqref{eq:noise_opt_nc}}{\leq} \frac{32 k }{c^2} \sum_{j = 0}^{k - 1} \norm{D^{(t + j)}}_F^2  + \frac{32k n \sigma^2}{c^2}
\end{align*}
Summing up $T_1$ and $T_2$ and estimating $k \leq 2 \tau$ we conclude the proof %
\begin{align*}
\E \norm{\Psi_{t + k} }_F^2 &\leq \frac{3}{4} \E \norm{\Psi_{t}}_F^2 +  \gamma^2  \frac{440 \tau}{c^2} \sum_{j = 0}^{k } \E\norm{\nabla f(X^{(t +j)}) -\nabla f(X^\star)}_F^2 + \gamma^2 \frac{380 \tau}{c^2} n \sigma^2  \,. \qedhere
\end{align*}
\end{proof}

We will proof Lemma~\ref{lem:consensus} with $B_1 = 28 C_1$, $B_2 = 4 C_2$, $B_3 = \sqrt{515 \cdot 2 C_1}$, where $C_1 = 220$ and $C_2 = 190$ are constants from Lemma~\ref{lem-aux}.
\begin{proof}[\bf Proof of Lemma~\ref{lem:consensus}]
	Observe, for any $t$,
	\begin{align*}
	\norm{\nabla f(X^{(t)}) -\nabla f(X^\star)}_F^2 &\stackrel{\eqref{eq:norm_of_sum_of_two}}{\leq} 2 \norm{\nabla f(X^{(t)}) - \nabla f(\bar X^{(t)})}_F^2 + 2\norm{\nabla f(\bar X^{(t)}) -\nabla f(X^\star)}_F^2 \\
	&\stackrel{\eqref{eq:smooth_nc}}{\leq} 2L^2 \underbrace{\norm{X^{(t)} -  \bar X^{(t)}}_F^2 }_{\leq \norm{\Psi_t}_F^2} + 2\norm{\nabla f(\bar X^{(t)}) -\nabla f(X^\star)}_F^2 
	\end{align*}
	With Lemma~\ref{lem-aux} 
	\begin{align*}
	\E \norm{\Psi_{t + k} }_F^2 &\stackrel{\eqref{eq:y_1}}{\leq} \frac{3}{4} \E \norm{\Psi_{t}}_F^2 +  \gamma^2\frac{\tau C_1}{c^2} \sum_{j = 0}^{k } \E \norm{\nabla f(X^{(t +j)})-\nabla f(X^\star)}_F^2 +  \gamma^2 \frac{\tau C_2}{c^2} n \sigma^2 \\
	& \leq \frac{3}{4} \E \norm{\Psi_{t}}_F^2 + \gamma^2 \frac{ 2C_1  \tau L^2}{c^2} \sum_{j = 0}^{k } \E \norm{\Psi_{t+j} }_F^2 + \gamma^2  \frac{2C_1 \tau}{c^2} \sum_{j = 0}^{k } \E \norm{ \nabla f(\bar X^{(t + j)})-\nabla f(X^\star)}_F^2+ \gamma^2 \frac{\tau C_2}{c^2} n \sigma^2\\
	& \stackrel{\gamma < \frac{c}{\sqrt{512 \cdot 2 C_1} L \tau}}{\leq} \frac{3}{4} \E\norm{\Psi_{t}}_F^2 + \frac{1}{512 \tau} \sum_{j = 0}^k \norm{\Psi_{t + j}}_F^2 + \gamma^2  \frac{2C_1 \tau}{c^2} \sum_{j = 0}^{k } \E \norm{ \nabla f(\bar X^{(t + j)})-\nabla f(X^\star)}_F^2+ \gamma^2 \frac{\tau C_2}{c^2} n \sigma^2
	\end{align*}
	Next, we estimate the third term by smoothness for $j < k$
	\begin{align*}
	\norm{ \nabla f(\bar X^{(t + j)})-\nabla f(X^\star)}_F^2 \leq 2 L n \left( f(\bar \xx^{(t + j)})  - f(\xx^\star) \right) \,.
	\end{align*}
	And for $j = k$, the index is $t + k$ and it should appear only in LHS. Thus we estimate
	\begin{align*}
	\norm{ \nabla f(\bar X^{(t + k)})-\nabla f(X^\star)}_F^2 & \stackrel{\eqref{eq:norm_of_sum_of_two}}{\leq} 2 \norm{\nabla f(\bar X^{(t + k)}) - \nabla f(\bar X^{(t + k - 1)}) }_F^2  + 2\norm{\nabla f(\bar X^{(t + k - 1)})  -\nabla f(X^\star)}_F^2 \\
	& \stackrel{\eqref{eq:smooth_nc}}{\leq }2 L^2 \norm{\bar X^{(t + k)} - \bar X^{(t + k - 1)}}_F^2 + 4 L n \left(f(\bar \xx^{(t + 
		k - 1)})  - f(\xx^\star)  \right)
	\end{align*}
	Next we use \eqref{eq:23}, that is equivalent to $\bar X^{(t + k)} = \bar X^{(t + k-1)} - \gamma \nabla F(X^{(t + k - 1)}, \xi^{(t + k - 1)}) \frac{\1\1^\top}{n}$. Taking expectation
	\begin{align*}
	\E \norm{\bar X^{(t + k)} - \bar X^{(t + k - 1)}}_F^2 &\stackrel{\eqref{eq:23}, \eqref{eq:noise_opt_nc}}{\leq } \gamma^2 \norm{\nabla f(X^{(t + k - 1)}) \textstyle{\frac{\1\1^\top}{n} } }_F^2 + \gamma^2 \sigma^2 \\
	& \leq 2 \gamma^2 \norm{\nabla f(X^{(t + k - 1)}) \textstyle{\frac{\1\1^\top}{n} } - \nabla \bar f(\bar X^{(t + k - 1)}) }_F^2 + 2 \gamma^2 \norm{\nabla \bar f(\bar X^{(t + k - 1)}) - \nabla \bar f( X^\star)}_F^2 + \gamma^2 \sigma^2 \\
	& \stackrel{\eqref{eq:norm_of_sum_of_two}, \eqref{eq:smooth_nc}}{\leq } 2 \gamma^2 L^2 \norm{X^{(t + k - 1)} - \bar X^{(t + k - 1)}}_F^2 +  4 \gamma^2 L n\left(f(\bar \xx^{(t + k - 1)})  - f(\xx^\star)  \right) + \gamma^2 \sigma^2 
	\end{align*}
	where on the second line we used $\nabla  f(\bar X) \frac{\1\1^\top}{n} = \nabla \bar f(\bar X)$, and  $\nabla \bar f( X^\star) = 0$. As $\gamma \leq \frac{c}{\sqrt{512 \cdot 2 C_1} L \tau}$
	\begin{align*}
	\norm{ \nabla f(\bar X^{(t + k)})-\nabla f(X^\star)}_F^2 \leq L^2 \norm{\Psi_{t + k - 1}}_F^2 + 5 L n  \left(f(\bar \xx^{(t + 
		k - 1)})  - f(\xx^\star)  \right)
	\end{align*}
	Coming back to recursion for $\norm{\Psi_{t + k}}_F^2$ and using that $\frac{2 C_1 L^2 \tau}{c^2} \gamma^2   \leq \frac{1}{512 \tau}$ by our choice of $\gamma$,
	\begin{align*}
	\E \norm{\Psi_{t + k} }_F^2 & \leq \frac{3}{4} \E \norm{\Psi_{t}}_F^2 + \frac{1}{256 \tau} \sum_{j = 0}^k \norm{\Psi_{t + j}}_F^2 +  \gamma^2  \frac{C_1 \tau}{c^2} 14 L n \sum_{j = 0}^{k - 1 } \E \left( f(\bar \xx^{(t + j)})  - f(\xx^\star) \right) + \gamma^2 \frac{2 C_2 \tau }{c^2} n \sigma^2
	\end{align*}
	It is only left to get rid of $\norm{\Psi_{t + k}}_F^2$ from RHS. For that we move the term with $\norm{\Psi_{t + k}}_F^2$ to LHS and divide the whole equation by $(1 - \frac{1}{256 \tau})$. We use that $\left(1 - \frac{1}{256 \tau}\right)^{-1} \leq 1 + \frac{1}{128 \tau} \leq 1 + \frac{1}{256} < 2$, and that $\left(1 - \frac{1}{4}\right) (1 + \frac{1}{128 \tau}) \leq \left(1 - \frac{1}{4}\right) (1 + \frac{1}{128}) \leq (1 - \frac{1}{8})$. We thus arrive to the Lemma's statement
	\begin{align*}
	\E \norm{\Psi_{t + k} }_F^2 & \leq \frac{7}{8} \norm{\Psi_{t}}_F^2 + \frac{1}{128 \tau} \sum_{j = 0}^{k - 1} \E  \norm{\Psi_{t + j}}_F^2 +  \gamma^2  \frac{28 C_1 \tau}{c^2} L n \sum_{j = 0}^{k - 1 } \E \left( f(\bar \xx^{(t + j)})  - f(\xx^\star) \right) + \gamma^2 \frac{4 C_2 \tau }{c^2} n \sigma^2
	\end{align*}
\end{proof}

\paragraph{Proof of Lemma~\ref{lem:unroll_rec}.}

\begin{proof}
	Define $\alpha = 28 C_1 \frac{\tau}{c^2} L n$, $\beta = 4 C_2 \frac{\tau}{c^2} \sigma^2 n$ for simplicity. Then inequality \eqref{eq:y} takes the form
\begin{align}\label{eq:inequality2}
\E \norm{\Psi_{t + k} }_F^2 \leq \left(1 - \frac{1}{8} \right) \E \norm{\Psi_{t} }_F^2 + \frac{1}{128 \tau} \sum_{j = 0}^{k - 1} \E \norm{\Psi_{t + j} }_F^2 + \alpha \gamma^2 \sum_{j = 0}^{k - 1} \E e_{t + j} + \beta \gamma^2 
\end{align}
\paragraph{A new quantity.} We define a new quantity that has non-increasing properties even for $k < \tau$ in contrast to $\E \norm{\Psi_{t + k} }_F^2$. For $t \geq 0$ we define
\begin{align*}
\Phi_{t + \tau} := \frac{1}{\tau}\sum_{j = 0}^{\tau - 1} \E \norm{\Psi_{t + j} }_F^2 && E_{t + \tau} := \alpha  \sum_{j = 0}^{\tau - 1} \E e_{t + j}
\end{align*}
\emph{Non-increasing property for $k < \tau$ (but $t + k \geq \tau$).}
\begin{align*}
\Phi_{t + k} = \frac{1}{\tau}\left(\sum_{i = k}^{\tau - 1}\E \norm{\Psi_{t - \tau + i} }_F^2 + \sum_{i = 0}^{k - 1}\E \norm{\Psi_{t + i} }_F^2\right)
\end{align*}
Applying \eqref{eq:inequality2} to the second sum,
\begin{align*}
\Phi_{t + k} &\leq \frac{1}{\tau} \sum_{i = k}^{\tau - 1} \E \norm{\Psi_{t - \tau + i} }_F^2 + \frac{1}{\tau} \sum_{i = 0}^{k - 1} \left[\left(1 - \frac{1}{8}\right)\E \norm{\Psi_{t - \tau + i} }_F^2 + \frac{1}{128} \Phi_{t + i} + \gamma^2 E_{t + i} + \beta \gamma^2 \right]
\end{align*}
\begin{align}\label{eq:second_recursion}
\Phi_{t + k} \leq \Phi_t + \frac{1}{128 \tau} \sum_{i = 0}^{k - 1} \Phi_{t + i} + \frac{1}{\tau}\gamma^2 \sum_{i = 0}^{k - 1} E_{t + i} + \frac{k}{\tau} \beta \gamma^2,
\end{align}
where we used that $\Theta_t \geq 0 ~~ \forall t$ and that $\tau \geq k$.

\emph{Contraction property for $ \tau \leq  k \leq 2 \tau$. } Using \eqref{eq:inequality2} and a definition of $\Phi_{t + k}$,
\begin{align*}
\Phi_{t + k} = \frac{1}{\tau}\sum_{j = k - \tau}^{k - 1}  \E \norm{\Psi_{t + j} }_F^2  \leq \left(1-\frac{1}{8}\right)\underbrace{\frac{1}{\tau} \sum_{j = k - \tau}^{k - 1}  \E \norm{\Psi_{t + j - \tau } }_F^2  }_{\Phi_{t + k - \tau}} + \frac{1}{128 \tau} \sum_{j = k - \tau}^{k - 1} \Phi_{t + j} + \gamma^2 \frac{1}{\tau} \sum_{i = k - \tau}^{k - 1} E_{t + i} + \beta \gamma^2 
\end{align*}
Combining with \eqref{eq:second_recursion} we get contraction for $\Phi_{t + k}$
\begin{align}\label{eq:first_recursion}
\Phi_{t + k}  &\stackrel{\eqref{eq:second_recursion}}{\leq}  \left(1-\frac{1}{8}\right) {\Phi_{t }} + \frac{1}{128 \tau} \sum_{j = 0}^{k - 1} \Phi_{t + j} + \gamma^2 \frac{1}{\tau} \sum_{i = 0 }^{k - 1} E_{t + i} + 2 \beta \gamma^2
\end{align}

\emph{Simplifying contraction property.} First, we substitute \eqref{eq:second_recursion} into the second term of \eqref{eq:first_recursion}
\begin{align*}
\Phi_{t + k} &\leq \left(1 - \frac{1}{8}\right) \Phi_t + \frac{1}{128 \tau} \sum_{i = 0}^{k - 2} \Phi_{t + i} + \frac{1}{128 \tau} \left[ \Phi_t + \frac{1}{128 \tau} \sum_{i = 0}^{k - 2} \Phi_{t + i} + \gamma^2 \frac{1}{\tau} \sum_{i = 0}^{k - 2} E_{t  + i} +  2 \beta \gamma^2  \right] + \gamma^2 \frac{1}{\tau} \sum_{i = 0}^{k - 1} E_{t + i} + 2 \beta \gamma^2 \\
&\leq \left(1 - \frac{1}{8}\right) \left(1 + \frac{1}{64 \tau}\right) \Phi_t +\left(1 + \frac{1}{128 \tau}\right) \left[ \frac{1}{128 \tau}  \sum_{i =0}^{k - 2} \Phi_{t + i} + \gamma^2 \frac{1}{\tau}\sum_{i = 0}^{k - 2} E_{t + i} + 2 \beta \gamma^2 \right]+ \gamma^2 \frac{1}{\tau} E_{t + k - 1}
\end{align*}
where we used that $\frac{1}{128 \tau} = \left(1 - \frac{1}{2}\right) \frac{1}{64 \tau} \leq \left(1 - \frac{1}{8}\right) \frac{1}{64 \tau} $. Similarly applying \eqref{eq:second_recursion} to the rest of $\Phi_{t + i}$,
\begin{align*}
\Phi_{t + k} &\leq \left(1 - \frac{1}{8}\right) \left(1 + \frac{1}{64 \tau}\right)^k \Phi_t + \gamma^2 \frac{1}{\tau} \sum_{i = 0}^{k - 1} \left(1 + \frac{1}{128 \tau}\right)^{t + k - 1 - i}  E_{t + i} + \left(1 + \frac{1}{128\tau} \right)^k  2 \beta \gamma^2
\end{align*}
We further use $\left(1 + \frac{1}{64 \tau}\right)^k \leq \left(1 + \frac{1}{64 \tau}\right)^{2 \tau} \leq \exp(\frac{1}{32}) \leq 1 + \frac{1}{16} $ and $\left( 1 - \frac{1}{8} \right)\left(1 + \frac{1}{64 \tau} \right)^{k} \leq \left(1 - \frac{1}{16}\right)$; and that $\left(1 + \frac{1}{128\tau} \right)^k \leq  1 + \frac{1}{32} \leq 2$.
Therefore,
\begin{align}\label{eq:inequality_phi}
\Phi_{t + k} &\leq \left(1 - \frac{1}{16}\right) \Phi_t + 2 \gamma^2 \frac{1}{\tau}\sum_{i = 0}^{k - 1} E_{t + i} + 4 \beta \gamma^2 
\end{align}

\emph{Simplifying non-increasing property \eqref{eq:second_recursion}.} 
Similarly as above we substitute recursively \eqref{eq:second_recursion} into the second term of \eqref{eq:second_recursion}, for $0<k < \tau$
\begin{align*}
\Phi_{t + k} &\leq \left(1 + \frac{1}{128 \tau}\right) \Phi_{t} + \left(1 + \frac{1}{128 \tau}\right)  \left[\frac{1}{128\tau} \sum_{i = 0}^{k - 2}\Phi_{t + i} +  \gamma^2 \frac{1}{\tau}\sum_{i = 0}^{k - 2} E_{t + i} + \beta\gamma^2  \right]+ \gamma^2 \frac{1}{\tau} E_{t + \tau - 1} \\
& \leq \left(1 + \frac{1}{128\tau}\right)^\tau \Phi_t  + \gamma^2 \frac{1}{\tau} \sum_{i = 0}^{\tau - 1} \left(1 + \frac{1}{128 \tau}\right)^{t + \tau - 1 - i}  E_{t + i} + \left(1 + \frac{1}{128\tau} \right)^\tau \beta \gamma^2
\end{align*}
Using now that $\left(1 + \frac{1}{128\tau} \right)^\tau \leq 2$ we get
\begin{align}\label{eq:second_recursion_phi}
\Phi_{t + k} \leq 2 \Phi_t + 2 \gamma^2 \frac{1}{\tau}\sum_{i = 0}^{\tau - 1} E_{t + i} + 2 \beta \gamma^2 
\end{align}

\paragraph{Obtaining recursion for $\E \norm{\Psi_{t} }_F^2 + \Phi_t$. } As our final goal is to obtain inequality for $\E \norm{\Psi_{t} }_F^2$, we start modifying \eqref{eq:inequality2}, for $\tau \leq k \leq 2 \tau$
\begin{align*}
\E \norm{\Psi_{t + k} }_F^2 &\leq \left(1 - \frac{1}{8} \right) \E \norm{\Psi_{t} }_F^2 + \frac{1}{128 } \left( \Phi_{t + k} + \Phi_{t + \tau} \right) + \alpha \gamma^2 \sum_{j = 0}^{k - 1} e_{t + j} + \beta \gamma^2 \\
& \stackrel{\eqref{eq:second_recursion_phi}}{\leq}  \left(1 - \frac{1}{8} \right) \E \norm{\Psi_{t} }_F^2  + \frac{1}{128 } \left[  4 \Phi_t +4 \gamma^2 \frac{1}{\tau} \sum_{j = 0}^{k - 1} E_{t + j }  \right] + \alpha \gamma^2 \sum_{j = 0}^{k - 1} e_{t + j} + 2 \beta \gamma^2 \\
& \leq \left(1 - \frac{1}{8} \right) \E \norm{\Psi_{t} }_F^2  + \frac{1}{32}\Phi_t + \frac{\gamma^2}{32} \frac{1}{\tau} \sum_{j = 0}^{ k - 1} E_{t + j} + \alpha \gamma^2 \sum_{j = 0}^{k - 1} e_{t + j} + 2 \beta \gamma^2 
\end{align*}
Summing up the last inequality and \eqref{eq:inequality_phi} we get 
\begin{align*}
\E \norm{\Psi_{t + k} }_F^2  + \Phi_{t + k} &\leq \left(1 - \frac{1}{32}\right) \left[\E \norm{\Psi_{t } }_F^2  + \Phi_t \right] + 3 \gamma^2 \frac{1}{\tau}\sum_{j = 0}^{k - 1} E_{t + j} + \gamma^2 \alpha \sum_{j = 0}^{k - 1} e_{t + j} + 6 \beta \gamma^2 
\end{align*}

\emph{Unrolling recursion up to $\tau$.} 
For a given $t \geq \tau$, lets define $m = \lfloor t / \tau\rfloor - 1$. Then 
\begin{align*}
\E \norm{\Psi_{t } }_F^2+ \Phi_t \leq \left(1 - \frac{1}{32}\right) \left[\E \norm{\Psi_{m \tau} }_F^2+ \Phi_{m\tau}\right]+ 3 \gamma^2 \frac{1}{\tau}\sum_{j = m\tau}^{t - 1}E_{j} + \gamma^2 \alpha \sum_{j = m\tau}^{t - 1} e_{j} + 6 \beta\gamma^2 
\end{align*}
Unrolling this recursively up to $\tau$ we get,
\begin{align}\label{eq:recursion}
\E \norm{\Psi_{t } }_F^2 + \Phi_t  &\leq \left(1 - \frac{1}{32}\right)^{m} \left[\E \norm{\Psi_{\tau } }_F^2+ \Phi_\tau\right]+  \gamma^2 \sum_{j = \tau}^{t - 1} \left(1 - \frac{1}{32}\right)^{\lfloor(t - j) / \tau \rfloor} \left[3 \frac{1}{\tau} E_j + \alpha e_j \right]+ 6 \beta \gamma^2 \sum_{j = 0}^{m - 1} \left(1 - \frac{1}{32}\right)^{j} 
\end{align}

\paragraph{Initial conditions. }
Inequality above work for $t \geq \tau$. Here, we focus on $t < \tau$. Using similar calculations as in Lemma~\ref{lem:consensus} replacing estimation of $\norm{\Psi^{0} J^t}_F^2$ by Lemma~\ref{lem:norm_J}, we get that
\begin{align}\label{eq:initial_cond}
\E \norm{\Psi_{t } }_F^2 \leq \underbrace{2 \norm{\Delta X^{(0)}}_F^2 + \frac{3 \gamma^2}{p^2} \norm{\Delta Y^{(0)}}^2_F}_{:= \tilde{\Theta}_0} + \frac{1}{128 \tau} \sum_{j = 0}^{t - 1} \E \norm{\Psi_{j} }_F^2 + \alpha \gamma^2 \sum_{j = 0}^{t - 1} e_{j} + \beta \gamma^2 
\end{align}

Recursively applying \eqref{eq:initial_cond} to the second term of \eqref{eq:initial_cond}, similarly as above, we get
\begin{align}\label{eq:initial_fin}
\E \norm{\Psi_{t } }_F^2\leq 2 \tilde\Theta_0 + 2 \alpha \gamma^2 \frac{1}{\tau}\sum_{j = 0}^{t - 1} e_{j} + 2 \beta \gamma^2 
\end{align}
And therefore,
\begin{align}\label{eq:initial_fin2}
\Phi_\tau = \frac{1}{\tau}\sum_{j = 0}^{\tau - 1} \E \norm{\Psi_{ j} }_F^2\leq 2  \tilde\Theta_0 + 2 \alpha \gamma^2 \frac{1}{\tau}\sum_{j = 0}^{\tau - 1} e_{j} + 2 \beta \gamma^2 
\end{align}

\paragraph{Final recursion. }
Finally we apply \eqref{eq:initial_fin}, \eqref{eq:initial_fin2} to the first term of \eqref{eq:recursion}
\begin{align*}
\E \norm{\Psi_{t } }_F^2 + \Phi_t  &\leq \left(1 - \frac{1}{32}\right)^{m} 4 \tilde \Theta_0 +  \gamma^2 \sum_{j = \tau}^{t - 1} \left(1 - \frac{1}{32}\right)^{\lfloor(t - j) / \tau \rfloor} \left[3 \frac{1}{\tau} E_j + 5 \alpha e_j \right] + 10 \beta \gamma^2 \sum_{j = 0}^{m - 1} \left(1 - \frac{1}{32}\right)^{j} 
\end{align*}
\begin{itemize}
	\item For the last term we estimate $\sum_{j = 0}^{m - 1} \left(1 - \frac{1}{32}\right)^{j}  \leq 2$.
	\item For the terms with $e_j$ and $E_j$ we estimate, similar to \cite{koloskova2020unified}, \\
	\begin{align*}
	\left( 1 - \frac{1}{32}\right)^{1/\tau} &\leq  \exp(-\frac{1}{32\tau}) \leq 1- \frac{1}{64\tau} \qquad \qquad \text{and thus} \\
	\left(1 - \frac{1}{32}\right)^{\lfloor(t - j) / \tau \rfloor} &\leq \left( 1 - \frac{1}{64\tau} \right)^{\tau \lfloor ( t - j ) / \tau \rfloor} \leq \left( 1 - \frac{1}{64\tau} \right)^{ t - j }\left( 1 - \frac{1}{64\tau} \right)^{ -\tau } \leq 2 \left(1 - \frac{1}{64 \tau}\right)^{t - j}
	\end{align*}
	where as $\frac{1}{64\tau} \leq \frac{1}{2}$ we estimated $\left( 1 - \frac{1}{64\tau} \right)^{ -\tau } \leq \left(\frac{1}{1 - \frac{1}{64\tau}}\right)^\tau \leq ( 1 + \frac{1}{32\tau})^\tau\leq \exp(\frac{1}{32})  < 2$.
	\item Similarly, for $\tilde \Theta_0$ term we estimate \\
	$ \left(1 - \frac{1}{32}\right)^{m } = \left(1 - \frac{1}{32}\right)^{\lfloor\frac{t - \tau}{\tau}  \rfloor } \leq \left(1 - \frac{1}{64\tau}\right)^{ t } \left(1 - \frac{1}{64\tau}\right)^{ -2  \tau} \leq  2 \left(1 - \frac{1}{64\tau}\right)^{ t }$.
	\item For the terms with $E_j$ we additionally estimate
	\begin{align*}
	2 \left(1 - \frac{1}{64 \tau}\right)^{t - j} E_j  = 2 \left(1 - \frac{1}{64 \tau}\right)^{t - j} \sum_{i = j - \tau}^{j - 1} e_i =  2  \sum_{i = j - \tau}^{j - 1} \left(1 - \frac{1}{64 \tau}\right)^{t - i}  \left(1 - \frac{1}{64 \tau}\right)^{i - j} e_i 
	\end{align*}
	Further, $-\tau < i - j < -1$, and thus $\left(1 - \frac{1}{64 \tau}\right)^{i - j} \leq 2$ for all such $-\tau < i - j < -1$.
\end{itemize}
Therefore we obtain
\begin{align*}
\E \norm{\Psi_{t } }_F^2 + \Phi_t  &\leq \left(1 - \frac{1}{64\tau}\right)^{ t } 8 \tilde\Theta_0 +  22 \gamma^2 \alpha \sum_{j = 0}^{t - 1} \left(1 - \frac{1}{64\tau}\right)^{t - j} e_j + 20 \beta \gamma^2
\end{align*}
This brings us to the statement of the lemma.
\end{proof}
The rest of the proof follows closely \cite{koloskova2020unified}. 
\subsubsection{$\tau$-slow Sequences}

\begin{definition}[$\tau$-slow sequences \cite{StichK19delays}]\label{def:tau-slow}
	The sequence $\{a_t\}_{t \geq 0}$ of positive values is \emph{$\tau$-slow decreasing} for parameter $\tau > 0$ if
	\begin{align*}
	a_{t + 1} \leq a_t,\quad  \forall t \geq 0 && \text{and}, && a_{t + 1}\left( 1 + \frac{1}{2\tau}\right) \geq a_t, \quad\forall t \geq 0 \,.
	\end{align*}
	The sequence $\{a_t\}_{t \geq 0}$ is \emph{$\tau$-slow increasing} if $\{a_t^{-1}\}_{t \geq 0}$ is $\tau$-slow decreasing.
\end{definition}
\begin{proposition}[Examples]\label{prop:examples-tau-slow}\hfill\null
	\begin{enumerate}
		\item The sequence $\{\eta_t^2\}_{t\geq 0}$ with $\eta_t = \frac{a}{b + t}$, $b \geq 32\tau$ is $4 \tau$-slow decreasing. 
		\item The sequence of constant stepsizes $\{\eta_t^2\}_{t\geq 0}$ with $\eta_t = \eta$ is $\tau$-slow decreasing for any $\tau$.
		\item The sequence $\{w_t \}_{t \geq 0}$ with $w_t = (b + t)^2$, $b \geq 84\tau$ is $8\tau$-slow increasing. 
		\item The sequence of constant weights $\{w_t\}_{t\geq 0}$ with $w_t = 1$ is $\tau$-slow increasing for any $\tau$.
	\end{enumerate}
\end{proposition}

\subsubsection{The Main Recursion}
\begin{lemma}[The main recursion] \label{lem:cons+descent} Let $\{w_t\}_{t \geq 0}$ be $64 \tau$-slow increasing sequence, $W_t = \frac{1}{T + 1} \sum_{t = 0}^T w_t$, with $\gamma \leq \frac{c}{582 C_1 \tau L }$ it holds that 
	\begin{align}\label{eq:main_rec}
	\sum_{t = 0}^T w_t \E \norm{\Psi_{t} }_F^2 &\leq  \sum_{t = 0}^T w_t  \left(1 - \frac{1}{64\tau}\right)^{ t } 8 \tilde \Theta_0 + \frac{n}{6 L}\sum_{t = 0}^T e_t w_t + 40 C_2 \frac{\tau}{c^2} \sigma^2 n \gamma^2 W_T,
	\end{align}
	where  $e_t = f(\bar{\xx}^{(t)}) - f^\star$, $\tilde \Theta_0 = 2 \norm{\Delta X^{(0)}}_F^2 + \frac{3 \gamma^2}{p^2} \norm{\Delta Y^{(0)}}^2_F$, $C_1 = 440, C_2 = 380$. 
\end{lemma}
\begin{proof}
We start by averaging \eqref{eq:unrolled_rec} with weights $w_t$. Define $W_T = \sum_{t = 0}^T w_t$, $\alpha = 28 C_1 \frac{\tau}{c^2} L n$, $\beta = 4 C_2 \frac{\tau}{c^2} \sigma^2 n$, 
\begin{align*}
\sum_{t = 0}^T w_t \E \norm{\Psi_{t} }_F^2 \leq  \sum_{t = 0}^T w_t \left(1 - \frac{1}{64\tau}\right)^{ t } 8 \tilde \Theta_0 + 22 \gamma^2 \alpha \underbrace{\sum_{t = 0}^T w_t \sum_{j = 0}^{t - 1} \left(1 - \frac{1}{64\tau}\right)^{t - j} e_j }_{:= T_1} + 20 \beta \gamma^2 W_T
\end{align*}
For the middle term $T_1$ we use that $w_t$ are $64 \tau$-slow increasing sequences, i.e. $w_t \leq w_j \left(1 + \frac{1}{128 \tau}\right)^{t - j}$, we get
\begin{align*}
T_1 &= \sum_{t = 0}^T \sum_{j = 0}^{t - 1} \left(1 - \frac{1}{64\tau}\right)^{t - j} \left(1 + \frac{1}{128 \tau}\right)^{t - j} e_j w_j \leq \sum_{t = 0}^T \sum_{j = 0}^{t - 1} \left(1 - \frac{1}{128\tau}\right)^{t - j} e_j w_j \\
&\leq \sum_{j = 0}^T e_j w_j \sum_{t = j + 1}^{T } \left(1 - \frac{1}{128\tau}\right)^{t - j} \leq \sum_{j = 0}^T e_j w_j \sum_{t = 0}^{\infty} \left(1 - \frac{1}{128\tau}\right)^{t - j} \leq 128 \tau \sum_{t = 0}^T e_t w_t
\end{align*}
Therefore, 
\begin{align*}
\sum_{t = 0}^T w_t \Theta_t \leq \sum_{t = 0}^T w_t  \left(1 - \frac{1}{64\tau}\right)^{ t } 8 \tilde\Theta_0 + 2816 \gamma^2 \alpha \tau \sum_{t = 0}^T e_t w_t + 20 \beta \gamma^2 W_T
\end{align*}
Now using that $\gamma \leq \frac{c}{582 C_1 \tau L }$ and that $\alpha = 20 C_1 \frac{\tau}{c^2} L n$, $\beta = 2 C_2 \frac{\tau}{c^2} \sigma^2 n$. 
\begin{align*}
\sum_{t = 0}^T w_t \Theta_t &\leq \sum_{t = 0}^T w_t \left(1 - \frac{1}{64\tau}\right)^{ t } 8\tilde \Theta_0 + \frac{n}{6 L}\sum_{t = 0}^T e_t w_t + 40 C_2 \frac{\tau}{c^2} \sigma^2 n \gamma^2 W_T
\end{align*}
\end{proof}

\subsubsection{Combining with the Descent Lemma~\ref{lem:descent}}

\begin{lemma}\label{lem:cons+descent2}
	Define $D = \frac{\sigma^2}{n}$, $a = \frac{\mu}{2}$, $A = 24 L \frac{1}{n} \tilde \Theta_0$, $\tilde \Theta_0 = 2 \norm{\Delta X^{(0)}}_F^2 + \frac{3 \gamma^2}{p^2} \norm{\Delta Y^{(0)}}^2_F$, $B = 120 C_2 L \frac{\tau}{c^2} \sigma^2$, $C_1 = 440, C_2 = 380$. Then with $\gamma \leq \frac{c}{582 C_1 \tau L }$ it holds that 
\begin{align}
\frac{1}{2 W_T} \sum_{t = 0}^T w_t e_t &\leq \frac{1}{W_T}\sum_{t = 0}^T \left( \frac{\left(1 - \gamma a \right)}{\gamma} w_t r_t  - \frac{w_t}{\gamma}r_{t + 1}\right) + D \gamma + \frac{A}{W_T} \sum_{t = 0}^T w_t \left(1 - \frac{1}{64\tau}\right)^t  + B \gamma^2 \label{eq:rec_convex}
\end{align}	
\end{lemma}
\begin{proof}
	First, define $W_T = \sum_{t = 0}^T w_t$, $r_t = {\norm{\bar{\xx}^{(t)} - \xx^\star}}^2$, $\Theta_t = \sum_{i = 1}^n \E  \big \|\bar{\xx}^{(t)} - \xx_i^{(t)} \big \|^2$. In this notation, \eqref{eq:x} writes as 
	\begin{align*}
	r_{t + 1} &\leq \left(1 - \dfrac{\gamma\mu}{2}\right) r_t + \dfrac{\gamma^2\sigma^2}{n} - \gamma e_t+ \gamma \dfrac{3 L }{n} \Theta_t,
	\end{align*}
	 We rearrange \eqref{eq:x} by multiplying by $w_t$ and dividing by $\gamma$
	 \begin{align*}
	 w_t e_t &\leq \frac{\left(1 - \frac{\gamma\mu}{2}\right)}{\gamma} w_t r_t  - \frac{w_t}{\gamma}r_{t + 1} + \dfrac{\sigma^2}{n} w_t \gamma + \dfrac{3 L }{n} w_t \Theta_t,
	 \end{align*}
	 Now summing up, dividing by $W_T$, using that $\Theta_t \leq \E \norm{\Psi_t}^2_F$, and using \eqref{eq:main_rec}
	 \begin{align*}
	 \frac{1}{W_T} \sum_{t = 0}^T w_t e_t &\leq \frac{1}{W_T} \sum_{t = 0}^T \left( \frac{\left(1 - \frac{\gamma\mu}{2}\right)}{\gamma} w_t r_t  - \frac{w_t}{\gamma}r_{t + 1} \right) + \frac{\sigma^2}{n} \gamma  + \frac{1}{W_T} \sum_{t = 0}^T w_t  24 L\left(1 - \frac{1}{64\tau}\right)^{ t } \frac{1}{n} \tilde\Theta_0 \\
	 & \qquad \qquad + \frac{1}{2} \frac{1}{W_T}\sum_{t = 0}^T e_t w_t + 120 C_2 L \frac{\tau}{c^2} \sigma^2 \gamma^2 
	 \end{align*}
	 Putting the fourth term to LHS we get the statement of the lemma. 
\end{proof}

Now similar to \cite[Lemma~15]{koloskova2020unified} we obtain the rates of Theorem~\ref{thm:GT-better-upper-bound} for the strongly convex case, and  similar to \cite[Lemma~16]{koloskova2020unified} for the weakly convex case.

\subsubsection{Strongly Convex Case}

\begin{lemma}\label{lem:rate_strongly_convex}
	If non-negative sequences $\{r_t\}_{t\geq 0}, \{e_t\}_{t \geq 0}$ satisfy \eqref{eq:rec_convex} for some constants $a~>~0,~ D, A, B~\geq~0$, then there exists a constant stepsize $\gamma < \frac{1}{b}$ with $b \geq 128 a \tau$ such that for weights $w_t = (1 - a \gamma)^{-(t + 1)}$ and $W_T := \sum_{t= 0}^T w_t$ it holds:
	\begin{align*}
	\frac{1}{2 W_T} \sum_{t= 0}^T e_t w_t + a r_{T+1} \leq \tilde{\cO}\left((r_0 + \nicefrac{A}{2 a}) b\exp\left[-\frac{a (T + 1)}{b}\right] + \frac{D }{aT} + \frac{B}{a^2T^2} \right),
	\end{align*}
	where $\tilde{\cO}$ hides polylogarithmic factors. %
\end{lemma}
\begin{proof}
	Starting from \eqref{eq:rec_convex} and using that that $\frac{w_t (1 - a \gamma)}{\gamma} = \frac{w_{t - 1}}{\gamma}$ we obtain a telescoping sum,
	\begin{align*}
	\frac{1}{2 W_T}\sum_{t = 0}^T w_t e_t \leq \frac{1}{W_T\gamma} \left((1 - a \gamma) w_0 r_0 - w_T  r_{T + 1} \right)  + D \gamma + B \gamma^2 + \frac{A}{W_T} \sum_{t = 0}^T w_t \left(1 - \frac{1}{64\tau}\right)^t \,,
	\end{align*}
	And hence, 
	\begin{align*}
	\frac{1}{2 W_T}\sum_{t = 0}^T w_t e_t + \frac{w_T  r_{T + 1}}{W_T\gamma} \leq \frac{ r_0 }{W_T\gamma} + D \gamma + B \gamma^2 + \frac{A}{W_T} \sum_{t = 0}^T w_t \left(1 - \frac{1}{64\tau}\right)^t \,,
	\end{align*}
	Now we estimate the last term. We use that $2 \gamma a \leq \frac{1}{64  \tau}$ and thus $\left(1 - \frac{1}{64\tau}\right)^t \leq  (1 - a \gamma)^{ 2t } $ 
	\begin{align*}
	\frac{1}{W_T} \sum_{t = 0}^T (1 - a \gamma)^{-(t + 1)} \left(1 - \frac{1}{64\tau}\right)^t &\leq 	\frac{1}{W_T} \sum_{t = 0}^T (1 - a \gamma)^{ t -1} \leq \frac{1}{W_T} \frac{1}{2 a \gamma}
	\end{align*}
	where we used that $\frac{1}{1 - a \gamma} \leq \frac{1}{2}$.  Thus,
	\begin{align*}
	\frac{1}{2 W_T}\sum_{t = 0}^T w_t e_t + \frac{w_T  r_{T + 1}}{W_T\gamma} \leq \frac{ 1}{W_T\gamma} \left(r_0  + \frac{A }{2a } \right)+ D \gamma + B \gamma^2 \,,
	\end{align*}
	
	Using that $W_T \leq \frac{w_T}{a \gamma} $ and $W_T \geq w_T = (1 - a \gamma)^{-(T + 1)}$ we can simplify
	\begin{align*}
	\frac{1}{2 W_T}\sum_{t = 0}^T w_t e_t + a  r_{T + 1}\leq (1 - a\gamma)^{T + 1} \frac{1 }{\gamma}\left(r_0  + \frac{A }{2a } \right)  + D \gamma + B \gamma^2 \leq \frac{1 }{\gamma}\left(r_0  + \frac{A }{2a } \right) \exp\left[-a\gamma (T + 1)\right] + D \gamma + B \gamma^2\,,
	\end{align*}
	Now lemma follows by tuning $\gamma$ the same way as in \cite{Stich19sgd}.
	\begin{itemize}
		\item  If $\frac{1}{b} \geq \frac{\ln(\max\{2, a^2 (r_0 + \frac{A}{2a}) T^2/D\})    }{aT}$  then we choose $\eta = \frac{\ln(\max\{2, a^2 (r_0 + \frac{A}{2a}) T^2/D\})    }{aT}$ and get that 
		\begin{align*}
		\tilde{\cO}&\left(a (r_0 + \nicefrac{A}{2a})  T\exp\left[-\ln(\max\{2, a^2 (r_0 + \nicefrac{A}{2a}) T^2/D\})  \right] \right) + \tilde{\cO}\left(\frac{D}{aT} \right) + \tilde{\cO}\left(\frac{B}{a^2 T^2} \right)\\
		 &= \tilde{\cO}\left(\frac{D}{aT} \right) + \tilde{\cO}\left(\frac{B}{a^2 T^2} \right) \,,
		\end{align*}
		\item Otherwise $ \frac{1}{b} \leq \frac{\ln(\max\{2, a^2 (r_0 + \frac{A}{2a}) T^2/D\})    }{aT}$ we pick $\eta = \frac{1}{b}$ and get that 
		\begin{align*}		
		&\tilde \cO \left((r_0 + \nicefrac{A}{2a}) b\exp\left[-\frac{a (T + 1)}{b}\right] + \frac{D }{b} + \frac{B}{b^2} \right) \\
		& \leq \tilde{\cO}\left((r_0 + \nicefrac{A}{2a}) b \exp\left[-\frac{a (T + 1)}{b}\right] + \frac{D}{aT} + \frac{B}{a^2T^2} \right)\,. \qedhere
		\end{align*}
	\end{itemize}
\end{proof}

\subsection{Weakly Convex and Non Convex Cases}
\begin{lemma}\label{lem:rate_weakly_convex}
	If non-negative sequences $\{r_t\}_{t\geq 0}, \{e_t\}_{t \geq 0}$ satisfy \eqref{eq:rec_convex} with $a=0,~ D, A, B~\geq~0$, then there exists a constant stepsize $\gamma < \frac{1}{b}$ with $b \geq 128 a \tau$ such that for weights $\{w_t = 1\}_{t \geq 0}$ it holds that:
	\begin{align*}
	\frac{1}{(T + 1)} \sum_{t= 0}^T e_t \leq \cO \left(2  \left(\frac{c r_0}{T + 1}\right)^{\frac{1}{2}} + 2 B^{1/3}\left(\frac{r_0}{T + 1}\right)^{\frac{2}{3}} + \frac{b r_0 + A \tau}{T + 1}\right).
	\end{align*}
\end{lemma}
\begin{proof}
	With $a = 0$, constant stepsizes $\eta_t = \eta$ and weights $\{w_t = 1\}_{t \geq 0}$ \eqref{eq:rec_convex} is equivalent to 
	\begin{align*}
	\frac{1}{2(T + 1) }\sum_{t = 0}^T e_t &\leq \frac{1}{(T + 1)\gamma}\sum_{t=0}^T \left( r_t - r_{t + 1}\right)  + D \gamma + B  \gamma^2 + \frac{A}{T + 1} \sum_{t = 0}^T \left(1 - \frac{1}{64\tau}\right)^t  \\
	&\leq \frac{r_0}{(T + 1)\gamma} + D \gamma + B \gamma^2 + \frac{64 A \tau}{T + 1} .
	\end{align*}
	To conclude the proof we tune the stepsize for the first three terms using Lemma~\ref{lem:tuning_stepsize}.
\end{proof}
\begin{lemma}[Tuning the stepsize]\label{lem:tuning_stepsize}
	For any parameters $r_0 \geq 0, b \geq 0, e \geq 0, d \geq 0$ there exists constant stepsize $\eta \leq \frac{1}{b}$ such that 
	\begin{align*}
	\Psi_T :=  \frac{r_0}{\gamma (T + 1)} + D \eta  + B \eta^2 \leq 2  \left(\frac{D r_0}{T + 1}\right)^{\frac{1}{2}} +  2 B^{1/3}\left(\frac{r_0}{T + 1}\right)^{\frac{2}{3}} + \frac{b r_0}{T + 1}
	\end{align*}
\end{lemma}
\begin{proof}
	Choosing $\eta = \min\left\{\left(\frac{r_0}{D(T + 1)}\right)^{\frac{1}{2}}, \left(\frac{r_0}{B(T + 1)}\right)^{\frac{1}{3}} , \frac{1}{b} \right\} \leq \frac{1}{b}$ we have three cases
	\begin{itemize}
		\item $\eta  = \frac{1}{b}$ and is smaller than both $\left(\frac{r_0}{D(T + 1)}\right)^{\frac{1}{2}}$ and $\left(\frac{r_0}{B(T + 1)}\right)^{\frac{1}{3}} $, then 
		\begin{align*}
		\Psi_T &\leq  \frac{b r_0}{T + 1} + \frac{D}{b} + \frac{B}{b^2} \leq \left(\frac{D r_0}{T + 1}\right)^{\frac{1}{2}} + \frac{b r_0}{T + 1} + B^{1/3}\left(\frac{r_0}{T + 1}\right)^{\frac{2}{3}}
		\end{align*}
		\item $\eta = \left(\frac{r_0}{D(T + 1)}\right)^{\frac{1}{2}} < \left(\frac{r_0}{B(T + 1)}\right)^{\frac{1}{3}} $, then
		\begin{align*}
		\Psi_T &\leq  2 \left(\frac{r_0D}{T + 1 }\right)^{\frac{1}{2}}   + B \left(\frac{r_0}{D(T + 1)}\right) \leq   2 \left(\frac{r_0D}{T + 1}\right)^{\frac{1}{2}} + B^{\frac{1}{3}} \left(\frac{r_0}{(T + 1)}\right)^{\frac{2}{3}},
		\end{align*}
		\item The last case, $\eta = \left(\frac{r_0}{B(T + 1)}\right)^{\frac{1}{3}} < \left(\frac{r_0}{D(T + 1)}\right)^{\frac{1}{2}} $
		\begin{align*}
		\Psi_T &\leq  2 B^{\frac{1}{3}} \left(\frac{r_0}{(T + 1)}\right)^{\frac{2}{3}} + D \left(\frac{r_0}{B(T + 1)}\right)^{\frac{1}{3}} \leq 2 B^{\frac{1}{3}} \left(\frac{r_0}{(T + 1)}\right)^{\frac{2}{3}} + \left(\frac{D r_0}{T + 1}\right)^{\frac{1}{2}} \,. \qedhere
		\end{align*}
	\end{itemize}
\end{proof}

\subsection{Non-convex Case}
First, we state the descent Lemma for non-convex cases. Due to Lemma~\ref{lem:average}, it holds that 
\begin{lemma}[Descent lemma for non-convex case, Lemma 11 from \cite{koloskova2020unified}]\label{lem:decrease_nc} Under Assumptions as in Theorem~\ref{thm:GT-better-upper-bound},
	the averages $\bar{\xx}^{(t)} := \frac{1}{n}\sum_{i=1}^n \xx_i^{(t)}$ of the iterates of Algorithm \ref{alg:gt} with the constant stepsize $\gamma < \frac{1}{4 L (M + 1)}$ satisfy 
	\begin{align}\label{eq:main_recursion_nc}
	\EE{t + 1}{f(\bar{\xx}^{(t + 1)})} &\leq f(\bar{\xx}^{(t)}) - \frac{\gamma}{4}\norm{\nabla f(\bar{\xx}^{(t)})}_2^2 + \frac{\gamma L^2}{n} \sum_{i = 1}^n \norm{\bar{\xx}^{(t)} -  \xx_i^{(t)}}_2^2 + \frac{L}{n} \gamma^2 \sigma^2.
	\end{align}
\end{lemma}

Similarly as for the convex cases we prove the following recursion
\begin{lemma}[Consensus distance recursion]\label{lem:consensus_nc} There are exists absolute constants $C_1, C_2 > 0$ such that 
	\begin{align}
	\E \norm{\Psi_{t + k} }_F^2 & \leq \frac{3}{4} \norm{\Psi_{t}}_F^2 + \frac{1}{128 \tau} \sum_{j = 0}^{k - 1} \E  \norm{\Psi_{t + j}}_F^2 +  C_1 \gamma^2  \tau n\sum_{j = 0}^{k - 1 } e_{t + j} + C_2 \gamma^2  \left(\frac{\tau n }{c^2}   + \tau^2\right) \sigma^2
	\end{align}
	where $e_j = \norm{\nabla f(\bar \xx^{(j)})}^2$, $\tau \leq k \leq  2 \tau$, $\tau = \frac{2}{p}\log\left(\frac{50}{p} (1+ \log \frac{1}{p})\right) + 1$, $p$ and $c$ are defined in \eqref{def:p}, $\Psi_t = \left( \Delta X^{(t)}, \gamma \Delta Y^{(t)}\right)$ and is defined in \eqref{eq:25}.
\end{lemma}
\begin{proof}
	The proof starts exactly the same as in the convex cases, Lemma~\ref{lem-aux}. The difference comes when estimating terms $T_1$ and $T_2$. 
	\paragraph{The second term $T_2$.} After splitting the stochastic noise, 
	\begin{align*}
	\E[T_2] &\leq 3 \E \norm{\sum_{j = 1}^{k } \left(\nabla f(X^{(t +j)}) - \nabla f(X^{(t + j - 1)})\right)\tilde W^{\tau - j}}_F^2 + 6 k n \sigma^2 \\
	& \stackrel{\eqref{eq:norm_of_sum}}{\leq} 3 k \sum_{j = 1}^k \E \norm{\nabla f(X^{(t +j)}) - \nabla f(X^{(t + j - 1)})}_F^2 + 6 k n \sigma^2 
	\end{align*}
	Estimating separately
	\begin{align*}
	\E \norm{\nabla f(X^{(t +j)}) - \nabla f(X^{(t + j - 1)})}_F^2 &\stackrel{\eqref{eq:norm_of_sum}}{\leq} 3 \E \norm{\nabla f(X^{(t +j)}) - \nabla f(\bar X^{(t +j)}) }_F^2  + 3 \norm{\nabla f(\bar X^{(t + j - 1)})- \nabla f(X^{(t + j - 1)}) }^2_F \\
	& \qquad + 3 \norm{\nabla f(\bar X^{(t +j)})- \nabla f(\bar X^{(t + j - 1)})}_F^2 \\ 
	& \stackrel{\eqref{eq:smooth_nc}}{\leq}  3 L^2 \E \norm{X^{(t +j)} - \bar X^{(t +j)}}_F^2 + 3 L^2 \norm{\bar X^{(t + j - 1)}- X^{(t + j - 1)}}^2_F \\
	& \qquad + 3 L^2 \norm{\bar X^{(t +j)} - \bar X^{(t + j - 1)}}_F^2
	\end{align*}
	And for the last term we estimate
	\begin{align*}
	\E \norm{\bar \xx^{(t +j)} - \bar \xx^{(t + j - 1)}}_2^2 &\leq \gamma^2 \norm{\frac{1}{n} \sum_{i = 1}^n \nabla f_i(\xx_i^{(t + j - 1)})}^2_2 + \gamma^2 \frac{\sigma^2}{n} \\
	& \leq  2 \gamma^2 \norm{\frac{1}{n} \sum_{i = 1}^n \nabla f_i(\xx_i^{(t + j - 1)}) - \frac{1}{n} \sum_{i = 1}^n \nabla f_i(\bar \xx^{(t + j - 1)}) }^2_2 + 2 \gamma^2 \norm{\nabla f(\bar \xx^{(t + j - 1)})}^2 + \gamma^2 \frac{\sigma^2}{n} \\
	& \leq 2 \gamma^2 L^2 \frac{1}{n} \sum_{i = 1}^n \norm{ \xx_i^{(t + j - 1)} - \bar \xx^{(t + j - 1)}}^2  + 2 \gamma^2\norm{\nabla f(\bar \xx^{(t + j - 1)})}^2   + \gamma^2 \frac{\sigma^2}{n}
	\end{align*}
	Thus, using that $\gamma < \frac{1}{24 L \tau}$, $k \leq 2 \tau$
	\begin{align*}
	\E[T_2] \leq \tau \sum_{j = 0}^{k - 1} n \E \norm{\nabla f(\bar \xx^{(t + j)})}^2   + 21 L^2 \tau\sum_{j = 0}^{k - 1} \E \norm{X^{(t +j)}  - \bar X^{(t +j)} }_F^2+ 7 \tau n \sigma^2\,.
	\end{align*}
	\paragraph{Term $T_1$. } Similarly, after separating the stochastic noise with $Z^{(t)}=G^{(t)}-\nabla f(X^{(t)})$, 
	\begin{align*}
	T_1 \stackrel{\eqref{eq:norm_of_sum_of_two}}{\leq} 2 \norm{\sum_{j = 1}^{k} \left[\nabla f ( X^{(t + j)}) - \nabla f (X^{(t + j - 1)})\right] (k - j) \tilde W^{k - j} }_F^2 + 2\norm{\sum_{j = 1}^{k} \left(Z^{(t + j)} - Z^{(t + j - 1)}\right) (k - j) \tilde W^{k - j} }_F^2. 
	\end{align*}
	We add and subtract $\nabla f(\bar X^{t + j}), \nabla f(\bar X^{t + j - 1})$ in the first term and denote $D^{(j)} = \nabla f(X^{(j)} ) -\nabla f(\bar X^{(j)})$. 
	\begin{align*}
	T_1 & \leq  4 \norm{\sum_{j = 1}^{k} \left(D^{(t + j)} - D^{(t + j - 1)}\right)(k - j) \tilde W^{k - j} }_F^2  + 4 \norm{\sum_{j = 1}^{k} \left[\nabla f(\bar X^{t + j}) - \nabla f(\bar X^{t + j - 1})\right](k - j) \tilde W^{k - j}}_F^2\\
	& \qquad + 2\norm{\sum_{j = 1}^{k} \left(Z^{(t + j)} - Z^{(t + j - 1)}\right) (k - j) \tilde W^{k - j} }_F^2. 
	\end{align*}
	Terms with $D$ and $Z$ we estimate exactly the same as in the convex case, thus getting 
	\begin{align*}
	\E [T_1] &\stackrel{\eqref{eq:noise_opt_nc}}{\leq} \frac{64 k }{c^2} \sum_{j = 0}^{k - 1} \norm{D^{(t + j)}}_F^2  + \frac{32k n \sigma^2}{c^2} + 4 \underbrace{ \norm{\sum_{j = 1}^{k} \left[\nabla f(\bar X^{t + j}) - \nabla f(\bar X^{t + j - 1})\right](k - j) \tilde W^{k - j}}_F^2}_{T_3}
	\end{align*}
	It is only left to estimate the last term. For that we use Lemma~\ref{lem:norm_estimate2}, and $\frac{1}{p} \leq \tau$ due to our choice of $\tau$, 
	\begin{align*}
	T_3 &\stackrel{\eqref{eq:norm_of_sum}}{\leq }  k \sum_{j = 1}^{k}\norm{ \left[\nabla f(\bar X^{t + j}) - \nabla f(\bar X^{t + j - 1})\right](k - j) \tilde W^{k - j}}_F^2 \stackrel{\text{L}.~\ref{lem:norm_estimate2}}{\leq } 4 k \tau^2 \sum_{j = 1}^{k}\norm{ \nabla f(\bar X^{t + j}) - \nabla f(\bar X^{t + j - 1})}_F^2 \\
	& \leq 4 k \tau^2 \gamma^2  \sum_{j = 1}^k \left[2  L^2 \norm{ X^{(t + j - 1)} - \bar X^{(t + j - 1)}}^2_F  + 2  n \norm{\nabla f(\bar \xx^{(t + j - 1)})}^2   + \sigma^2 \right]
	\end{align*}
	Where the last inequality was obtained while estimating Term $T_2$. Using that $k \leq 2 \tau$, $\gamma \leq \frac{1}{24 L \tau}$ and that  $\norm{D^{(t + j)}}_F^2 \leq L^2 \norm{X^{(t + j)} - \bar X^{(t + j)}  }_F^2$ by smoothness
	\begin{align*}
	\E [T_1] &\stackrel{\eqref{eq:noise_opt_nc}}{\leq} \frac{129 \tau }{c^2} L^2 \sum_{j = 0}^{k - 1} \norm{ X^{(t + j)} - \bar X^{(t + j )}}^2_F  + \tau \sum_{j = 0}^{k - 1} n \norm{\nabla f(\bar \xx^{(t + j)})}^2 + \left(\frac{64 \tau n }{c^2}   + \tau^2\right)\sigma^2 
	\end{align*}
	Summing $T_1$ and $T_2$ together, and using that $\gamma \leq \frac{c}{310 \tau L }$
	\begin{align*}
	\E \norm{\Psi_{t + k} }_F^2 & \leq \frac{3}{4} \norm{\Psi_{t}}_F^2 + \frac{1}{128 \tau} \sum_{j = 0}^{k - 1} \E  \norm{\Psi_{t + j}}_F^2 +  \gamma^2  10 \tau n\sum_{j = 0}^{k - 1 } \norm{\nabla f(\bar \xx^{(t + j)})}^2 + 5 \gamma^2  \left(\frac{64 \tau n }{c^2}   + \tau^2\right) \sigma^2
	\end{align*}
\end{proof}
Next, we unroll this recursion with Lemma~\ref{lem:unroll_rec}. 

For $\gamma < \frac{c}{\sqrt{7 B_1} L \tau} \leq \frac{1}{2L\tau}$, and with some positive absolute constants $B_1, B_2 > 0$ it holds,
\begin{align}\label{eq:unrolled_rec_nc}
\E \norm{\Psi_{t}}_F^2  &\leq \left(1 - \frac{1}{64\tau}\right)^{ t} A_0  +   B_1 \tau \gamma^2  \sum_{j = 0}^{t - 1} \left(1 - \frac{1}{64\tau}\right)^{t - j} n e_j + B_2 \gamma^2   \left(\frac{\tau n }{c^2}   + \tau^2\right) \sigma^2 
\end{align}
where $e_j = \norm{\nabla f(\bar \xx^{(j)})}^2$, $A_0 =16 \|\Delta X^{(0)}\|_F^2 + \frac{24 \gamma^2}{p^2} \|\Delta Y^{(0)}\|^2_F $.

The rest of proof consists of combining \eqref{eq:unrolled_rec_nc} with the descent lemma for non-convex case \eqref{eq:main_recursion_nc} in similar fashion as in Lemmas~\ref{lem:cons+descent}, \ref{lem:cons+descent2}; and further using Lemma~\ref{lem:rate_weakly_convex} to obtain the final rate. 

\section{Experimental Setup and Additional Plots} \label{sec:exp_details}

We illustrate the dependence of the convergence rate on the parameters $c$ and $p$.

In these experiments, we
vary $p$ and $c$ (by changing the mixing matrix) and measure the value of $f(\bar \xx^{(t)} ) - f^\star$ that GT reaches after a large number of steps $t$,
when using a constant stepsize $\gamma$
(chosen small enough so that none of the runs diverges).
According to our theoretical results, GT converges to the level $\cO\left( \frac{\gamma \sigma^2}{n}  + \frac{\gamma^2 \sigma^2}{pc^2}\right)$ in a linear number of steps (to reach higher accuracy, smaller stepsizes must be used).
Thus, for $n$ large enough, this term is dominated by $\cO\left(\frac{\gamma^2 \sigma^2}{pc^2}\right)$, which we aim to measure.
In all experiments we ensure that the first term is at least by order of magnitude smaller than the second by comparing the noise level with GT on a fully-connected topology.

\subsection{Problem Instances}
We used $n=300$, $d=100$.

\textbf{Setup A (Gaussian Noise).} 
We consider quadratic functions defined as $f_i(\xx) = \norm{\xx}^2$, and $\xx^{(0)}$ is randomly initialized from a normal distribution $\cN(0, 1)$.
We add artificially stochastic noise to gradients as $\nabla F_i(\xx, \xi) = \nabla f_i(\xx) + \xi$, where $\xi \sim \cN(0, \frac{\sigma^2}{d} I)$. 

\textbf{Setup B (Structured Noise).} We consider quadratic functions defined as $f_i(\xx) = \norm{\xx}^2$, and $\xx^{(0)}$ is randomly initialized from a normal distribution $\cN(0, 1)$.
We add artificially stochastic noise to gradients as $\nabla F(X, \xi) = \nabla f(X) + \diag(\xi) V$, where $\xi \sim \cN(0, \frac{\sigma^2}{d} I)$ is a $d$-dimensional Gaussian noise vector, $\diag(\xi)$ a matrix with $\xi$ on the diagonal, and 
$V \in \R^{d \times n}$ is a matrix with half of the rows equal to $\vv \in \R^n$, and half of the rows equal to $\uu \in \R^n$, where $\vv,\uu$ are eigenvectors of the mixing matrix, $W \vv = \lambda_n(W) \vv$, i.e.\ corresponding to the smallest eigenvalue of $W$, and  $W \uu = \lambda_2(W) \uu$, i.e.\ corresponding to the second largest eigenvalue of $W$.

This is motivated by the observations in Lemma~\ref{lem:norm_estimate}, where we noted that components in the eigenspace corresponding to the smallest eigenvalue of $W$ get amplified the most.

\subsection{Graph Topologies and Mixing Matrices}
\textbf{Interpolated Ring (between uniform weights and interpolate with a fully-connected topology).}
We consider the ring topology $W_{\rm ring}$ on $n$ nodes, where each node $i$ has self weight $w_{ii}=\frac{1}{3}$ and $w_{i,1+(i \mod n)} = w_{i,(i-2 \mod n) + 1} =\frac{1}{3}$ for its neighbors.
We interpolate this uniform weight ring topology with a fully-connected topology, $W_{\rm complete} = \frac{1}{n}\1 \1^\top$, 
that is, $W_\alpha := \alpha W_{\rm ring} + (1-\alpha) W_{\rm complete}$. %
The eigenvalues of $W_{\rm ring}$ are $\lambda(W_{\rm ring}) \in \left[-\frac{1}{3},1\right]$, and $\lambda(W_{\rm complete}) \in [0,1]$, and therefore $c$ of $W_\alpha$ is also a constant.

\textbf{Ring with smaller self weight.}
We consider the ring topology $W_w$ on $n$ nodes, where each node $i$ has self weight $w_{ii}=w \leq \frac{1}{3}$ and $w_{i,1+(i \mod n)} = w_{i,1+(i-2 \mod n)} =\frac{1-w}{2}$ for its neighbors. 
The eigenvalues of $W_w$ are $\lambda(W_w) \in \left[2w-1,1\right]$, and therefore $c$ can become small by choosing $w$ (note that the $\lambda_n(W_w)$, while decreasing for smaller $w$, is not equal to $2w-1$ in general, expect when $w=\frac{1}{3}$). We measure the exact value $\lambda_n(W_w)$ when reporting $c$ below.

\subsection{Additional Plots for Setup A}
\begin{figure}[H]
	\centering
	\vspace{-1em}
	\subfigure[\small constant $c$.]{
		\includegraphics[width=.315\textwidth,]{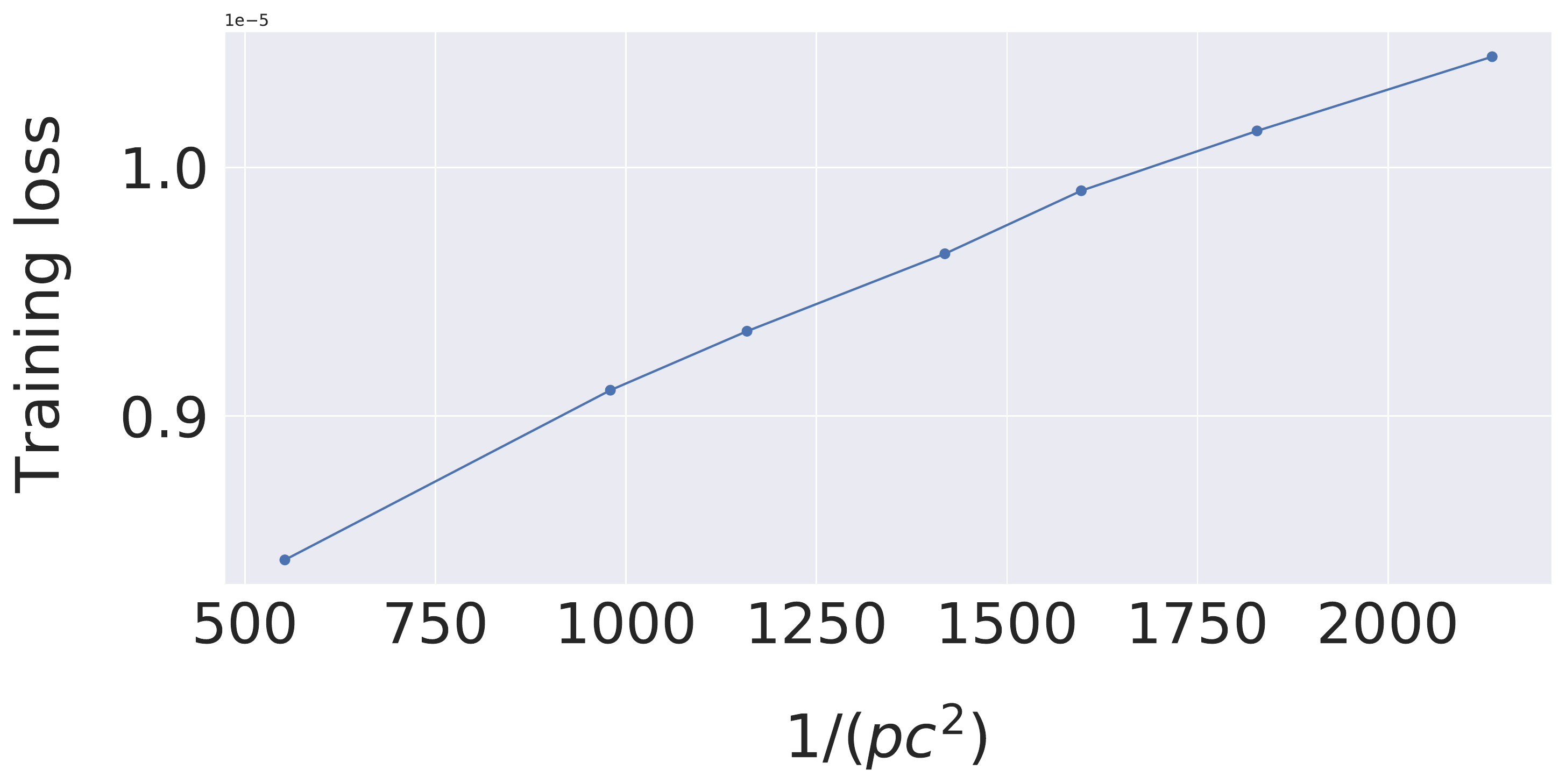}
		\label{fig:quadractics2_constant_c_inverse_c2p}
	}
	\hfill
	\subfigure[\small constant $c$.]{
		\includegraphics[width=.315\textwidth,]{figures/camera_ready/quadractics2_constant_c_inverse_p.pdf}
		\label{fig:quadractics2_constant_c_inverse_p}
	}
	\hfill
	\subfigure[\small constant $c$.]{
		\includegraphics[width=.315\textwidth,]{figures/camera_ready/quadractics2_constant_c_inverse_p2.pdf}
		\label{fig:quadractics2_constant_c_inverse_p2}
	}
	\hfill
	\subfigure[\small constant $c$.]{
		\includegraphics[width=.315\textwidth,]{figures/camera_ready/quadractics2_constant_c_inverse_p3.pdf}
		\label{fig:quadractics2_constant_c_inverse_p3}
	}
	\subfigure[\small constant $p$.]{
		\includegraphics[width=.315\textwidth,]{figures/camera_ready/quadractics2_constant_p_inverse_c2p.pdf}
		\label{fig:quadractics2_constant_p_inverse_c2p}
	}
	\subfigure[\small constant $p$.]{
		\includegraphics[width=.315\textwidth,]{figures/camera_ready/quadractics2_constant_p_inverse_cp.pdf}
		\label{fig:quadractics2_constant_p_inverse_cp}
	}
	\vspace{-1em}
	\caption{\small
	Impact of $c$ and $p$ on the convergence with the Gaussian stochastic noise $\sigma^2 = 1$. 
	The first four subfigures illustrate the impact of $p$ on convergence when $c$ is kept constant;
	showing a linear scaling of the loss compared to~$\frac{1}{p}$. 
	The last subfigure varies $c$ in the graph while keeping $p$ as a constant, and we see a linear scaling compared to $\frac{1}{c^2}$.	
	}
	\label{fig:additional_plot_setup_A}
\end{figure}

\subsection{Additional Plots for Setup B}
\begin{figure}[H]
	\centering
	\vspace{-1em}
	\subfigure[\small constant $c$.]{
		\includegraphics[width=.315\textwidth,]{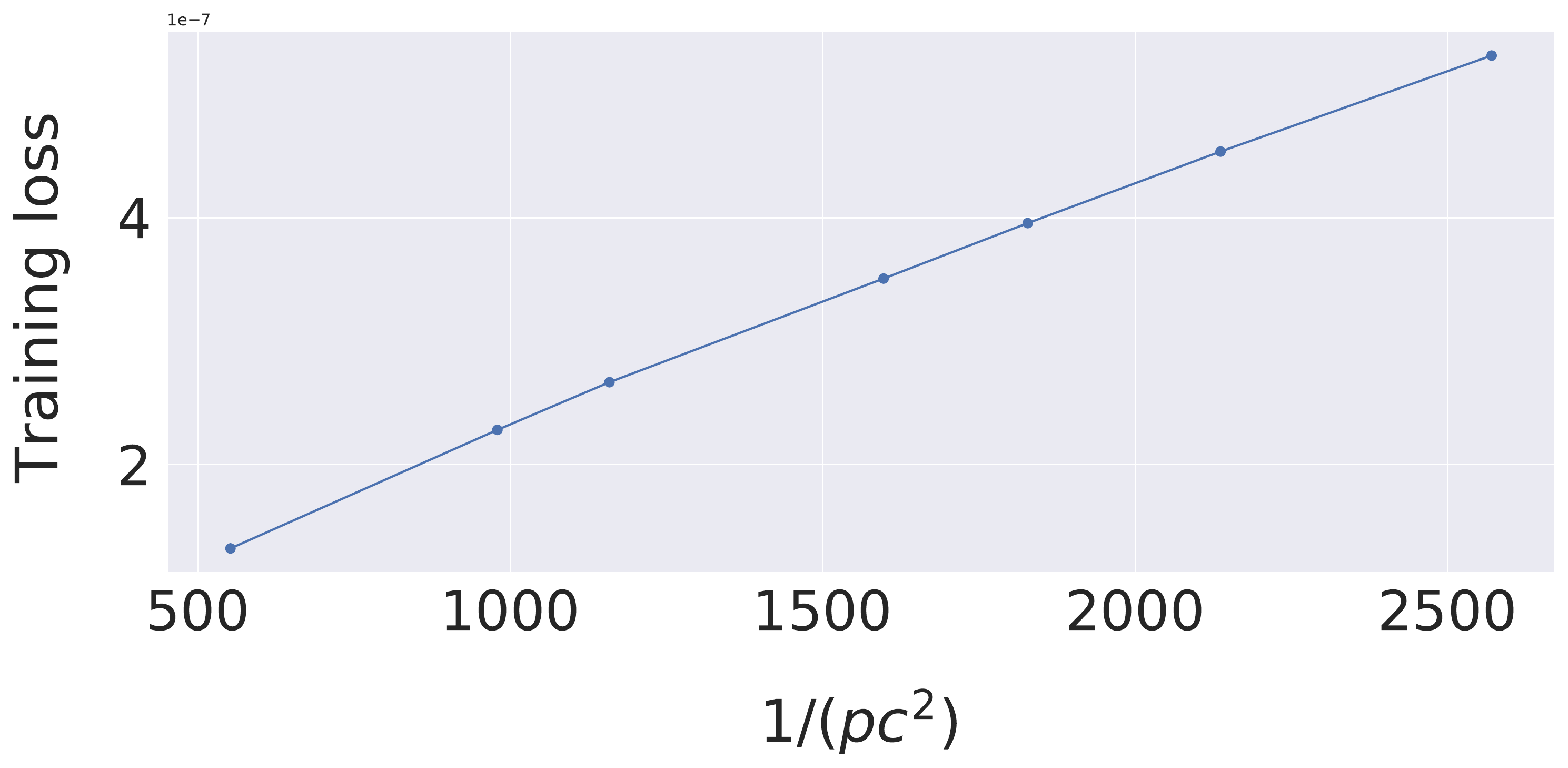}
		\label{fig:quadractics3_constant_c_inverse_c2p}
	}
	\hfill
	\subfigure[\small constant $c$.]{
		\includegraphics[width=.315\textwidth,]{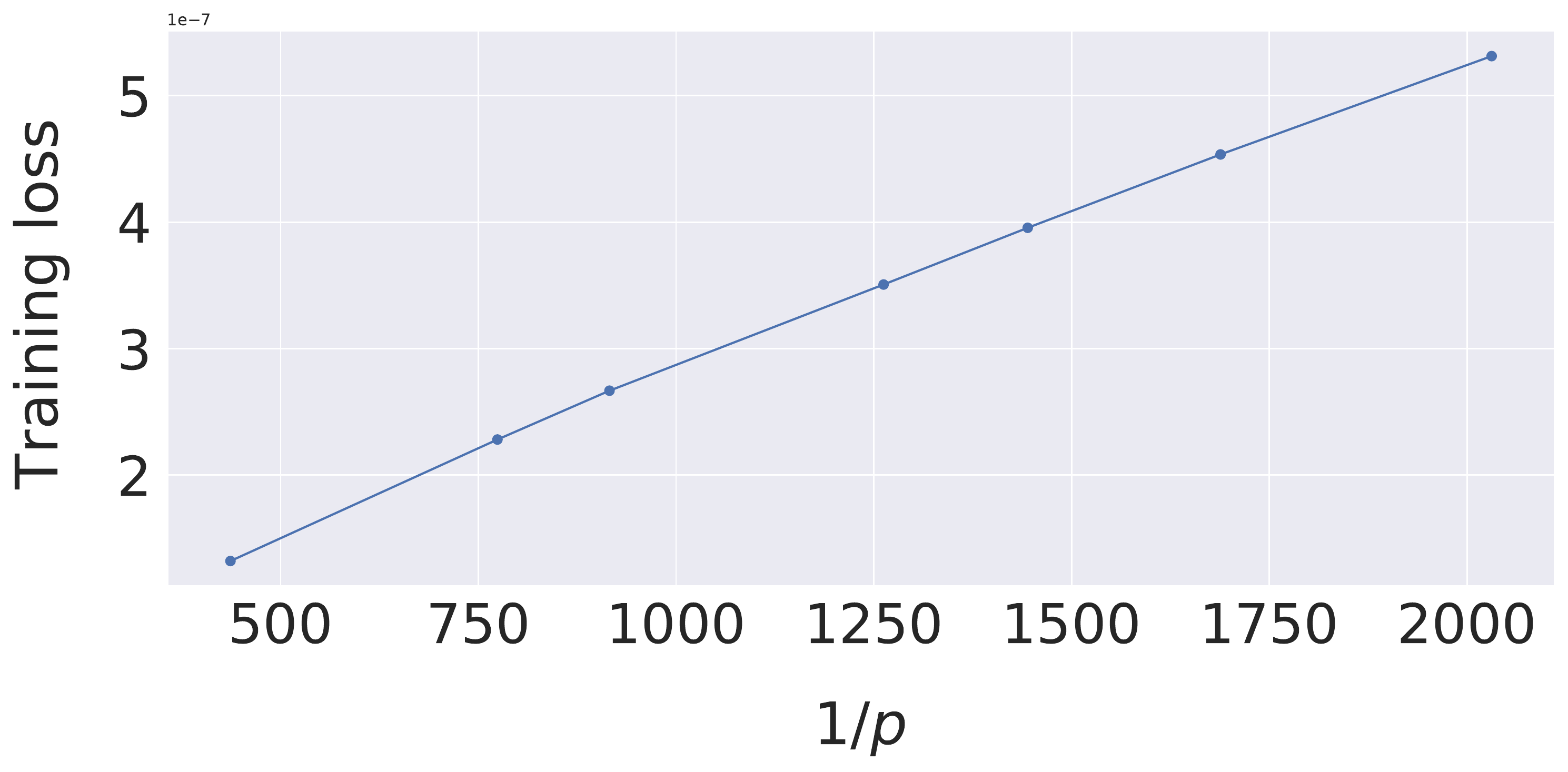}
		\label{fig:quadractics3_constant_c_inverse_p}
	}
	\hfill
	\subfigure[\small constant $c$.]{
		\includegraphics[width=.315\textwidth,]{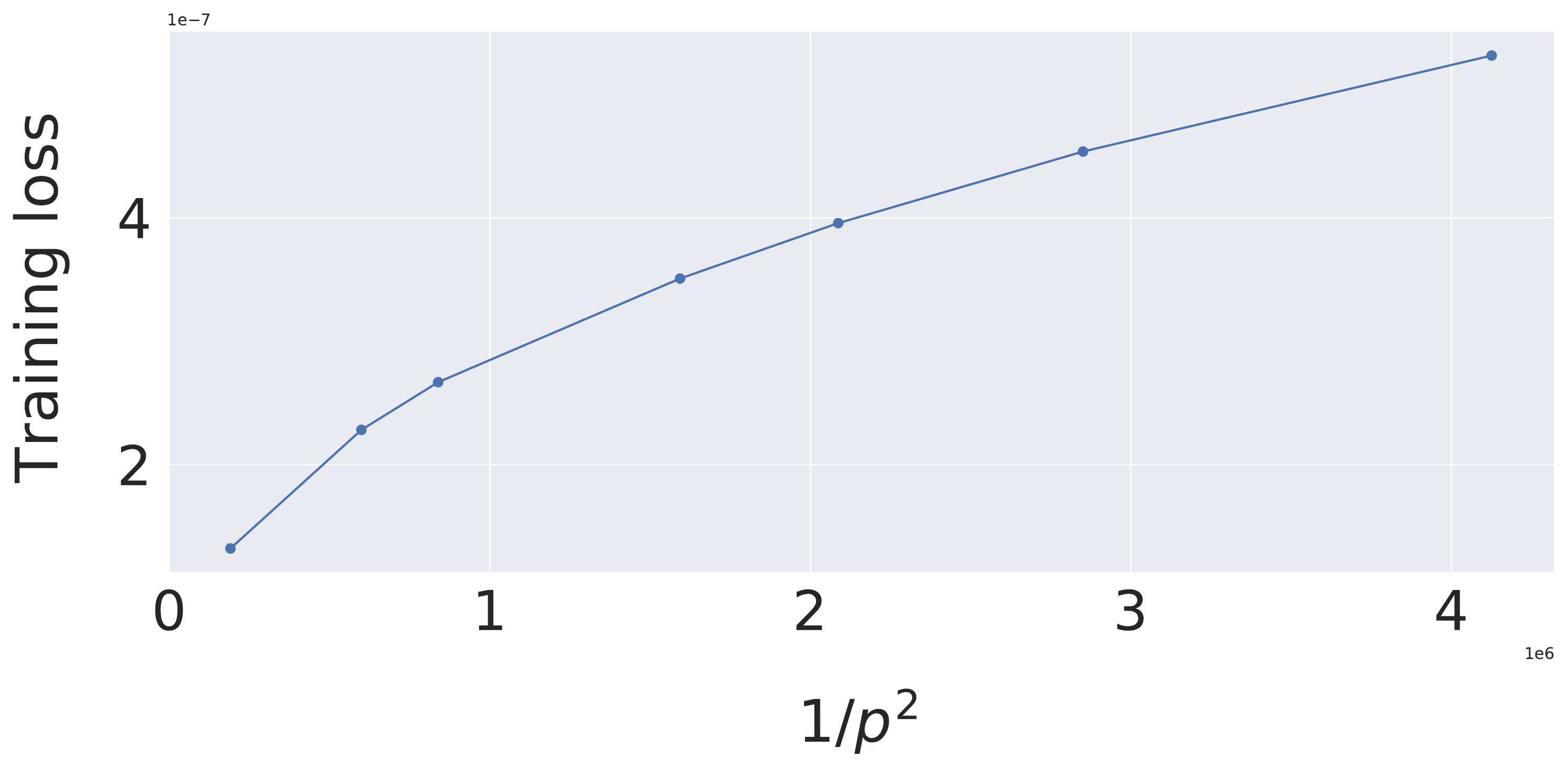}
		\label{fig:quadractics3_constant_c_inverse_p2}
	}
	\hfill
	\subfigure[\small constant $c$.]{
		\includegraphics[width=.315\textwidth,]{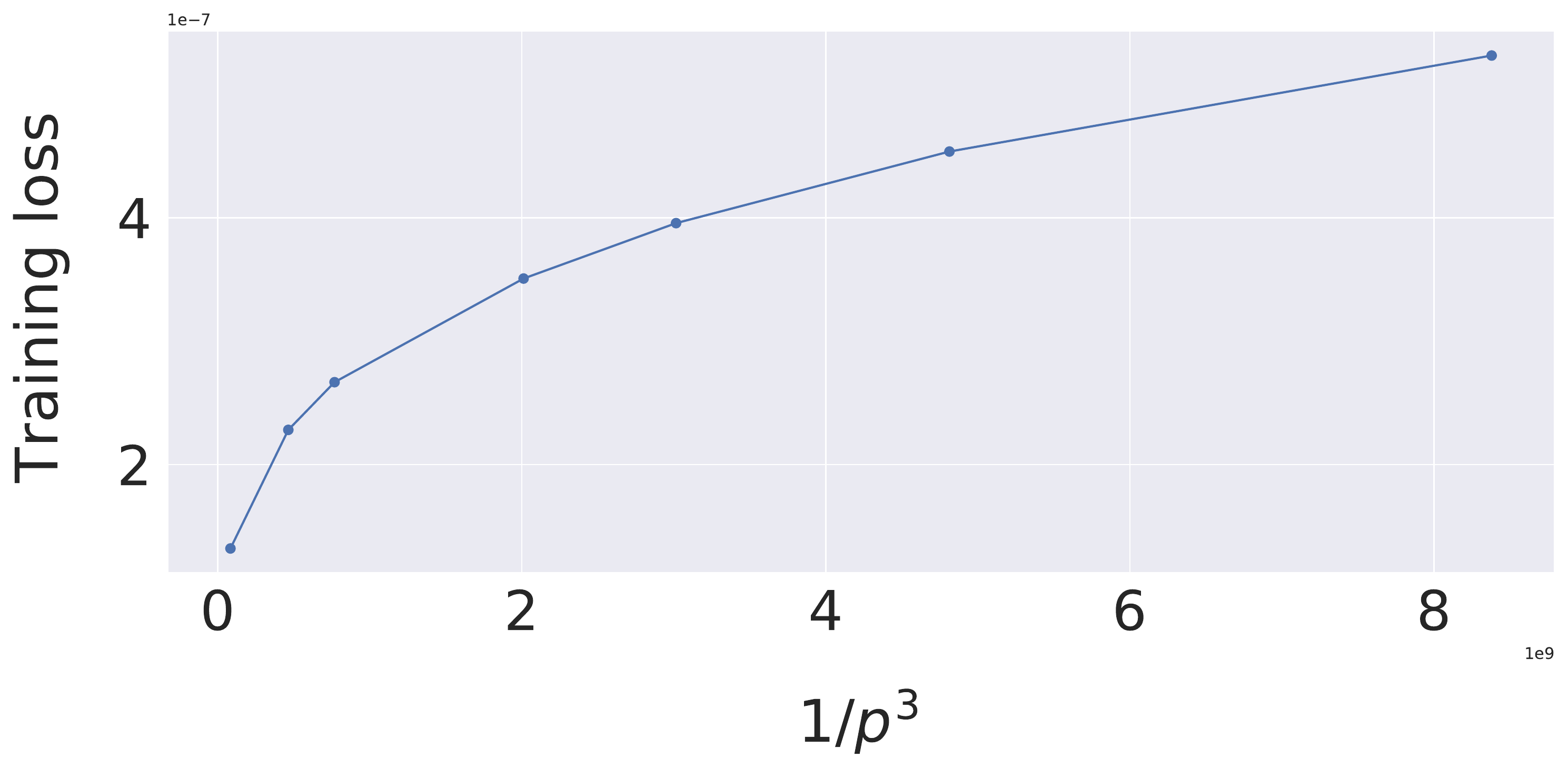}
		\label{fig:quadractics3_constant_c_inverse_p3}
	}
	\subfigure[\small constant $p$.]{
		\includegraphics[width=.315\textwidth,]{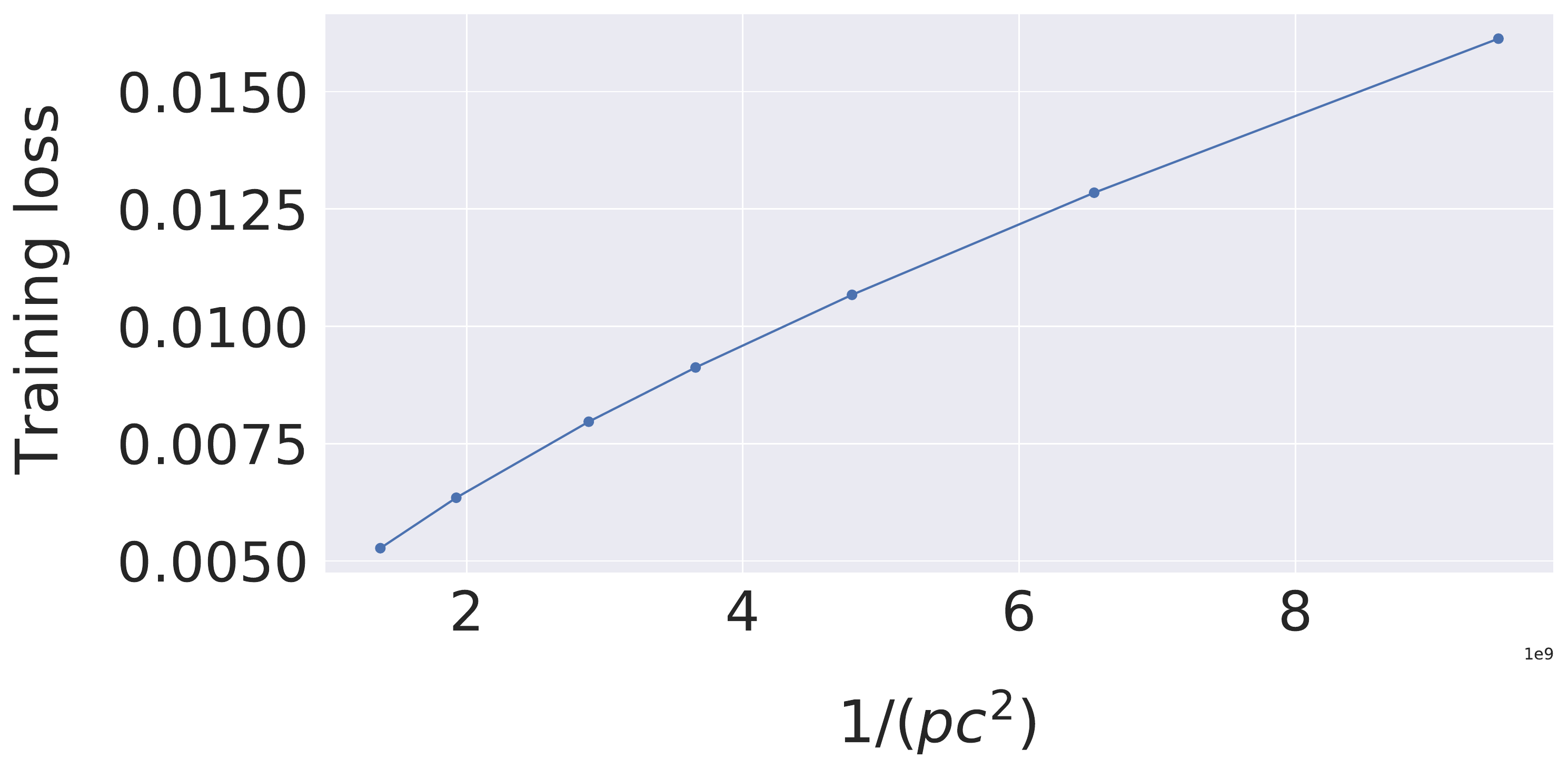}
		\label{fig:quadractics3_constant_p_inverse_c2p}
	}
	\subfigure[\small constant $p$.]{
		\includegraphics[width=.315\textwidth,]{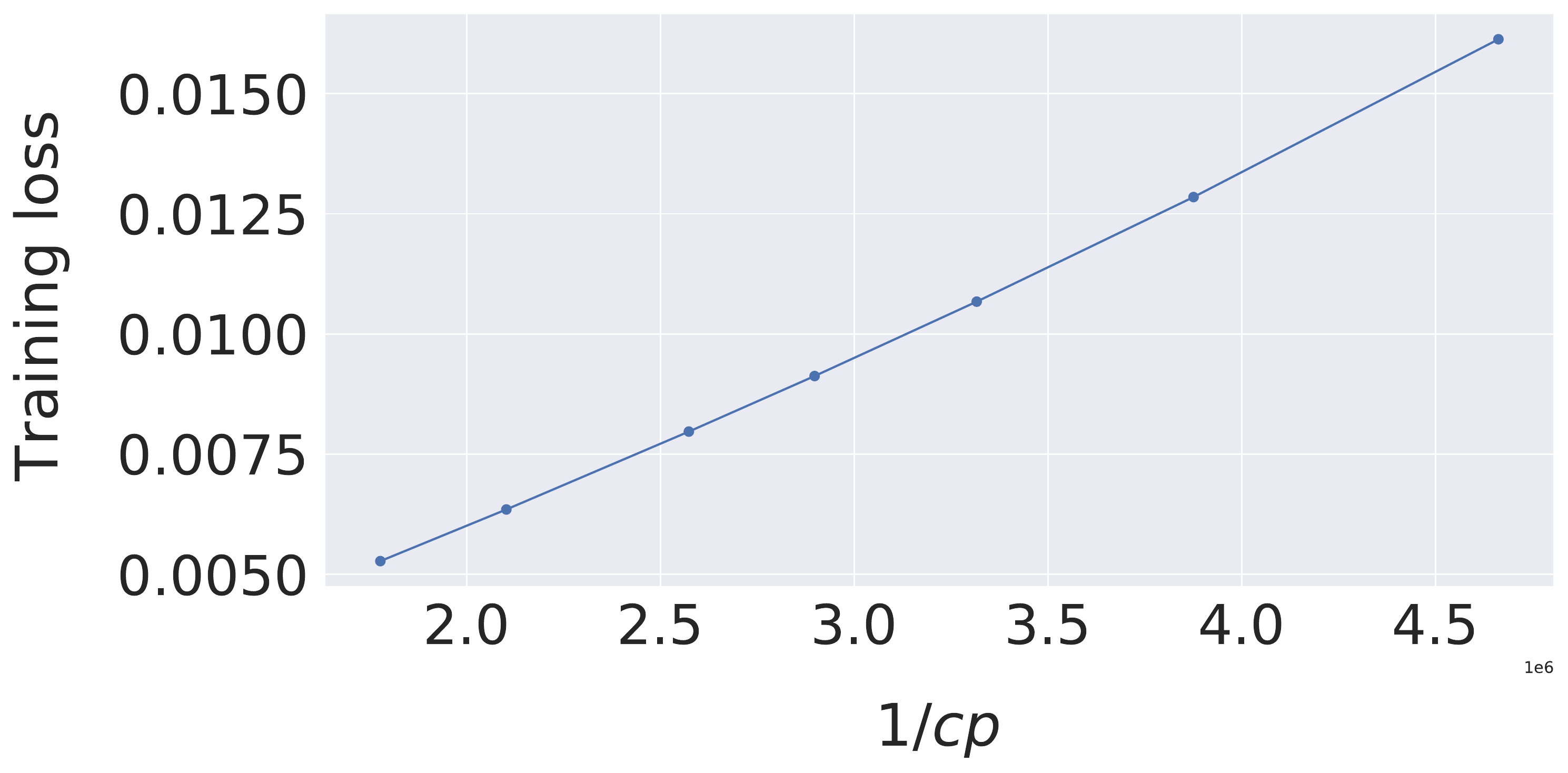}
		\label{fig:quadractics3_constant_p_inverse_cp}
	}
	\vspace{-1em}
	\caption{\small 
	Impact of $c$ and $p$ on convergence with the structured stochastic noise $\sigma^2 = 1$. 
	The first four subfigures illustrate the impact of $p$ on convergence when $c$ is kept constant;
	showing a linear scaling of the loss compared to~$\frac{1}{p}$. %
	The last subfigure varies $c$ in the graph while keeping $p$ as a constant, and we can see a linear scaling compared to $\frac{1}{c^2}$.
	}
	\label{fig:additional_plot_setup_B}
\end{figure}

In Figures~\ref{fig:additional_plot_setup_A} and~\ref{fig:additional_plot_setup_B}
we study the impact of $c$ and $p$ on the convergence. These findings support the $\cO\left( \frac{\gamma^2 \sigma^2}{pc^2}\right)$ scaling predicted by theory---however, cannot replace a formal proof.
We leave this for future work.